\documentclass[10pt,journal,compsoc]{IEEEtran}
\newtheorem{thm}{Theorem}[section]

\newtheorem{lem}[thm]{Lemma}

\newtheorem{defn}[thm]{Definition}

\newcommand{\K}{{\cal K}}

\newtheorem{theorem}{Theorem}

\newtheorem{remark}{Remark}

\usepackage{fancyhdr}

\usepackage{verbatim}
\usepackage{graphicx}
\usepackage[dvipsnames]{xcolor}
\usepackage{algorithm}
\usepackage{algorithmic}

\usepackage{amsmath, amssymb}
\usepackage{breqn}
\usepackage{threeparttable}

\usepackage{multirow}
\usepackage[numbers,sort&compress]{natbib}
\usepackage{subfig}
\usepackage{url}
\usepackage{hyperref}
\usepackage{breakurl}
\usepackage{booktabs}
\usepackage{color}
\usepackage{color}
\usepackage{framed}
\definecolor{shadecolor}{gray}{0.85}

\ifCLASSOPTIONcompsoc
  \usepackage[nocompress]{cite}
\else
  \usepackage{cite}
\fi
\ifCLASSINFOpdf
\else
\fi
\begin{document}
%
\title{Secure Hot Path Crowdsourcing with Local Differential Privacy under Fog Computing Architecture}
%
%
%
%

\author{Mengmeng Yang,~
        Ivan Tjuawinata,~
       Kwok Yan Lam,~\IEEEmembership{Senior Member,~IEEE,}
       Jun Zhao, ~\IEEEmembership{Member,~IEEE,}
        and, Lin Sun~
\IEEEcompsocitemizethanks{\IEEEcompsocthanksitem M. Yang, I. Tjuawinata, K. Y. Lam and J. Zhao are with Nanyang Technological University, Singapore. \protect\\
E-mail: $\{$melody.yang, ivan.tjuawinata, kwokyan.lam, and junzhao$\}$@ntu.edu.sg
\IEEEcompsocthanksitem L. Sun is with Tsinghua University, 100084, Beijing, China. \protect\\
E-mail: sunl16@mails.tsinghua.edu.cn
}
\thanks{Manuscript received April 19, 2020; revised August 26, 2020.}
\thanks{(Corresponding author: Jun Zhao)}
}

\IEEEtitleabstractindextext{%
\begin{abstract}
Crowdsourcing plays an essential role in the Internet of Things (IoT) for data collection, where a group of workers is equipped with Internet-connected geolocated devices to collect sensor data for marketing or research purpose.
In this paper, we consider crowdsourcing these worker's hot travel path. 
Each worker is required to report his real-time location information, which is sensitive and has to be protected. 
Encryption-based methods are the most direct way to protect the location, but not suitable for resource-limited devices. 
Besides, local differential privacy is a strong privacy concept and has been deployed in many software systems. 
However, the local differential privacy technology needs a large number of participants to ensure the accuracy of the estimation, which is not always the case for crowdsourcing.   
To solve this problem, we proposed a trie-based iterative statistic method, which combines additive secret sharing and local differential privacy technologies. 
The proposed method has excellent performance even with a limited number of participants without the need of complex computation. Specifically, the proposed method contains three main components: iterative statistics, adaptive sampling, and secure reporting. 
We theoretically analyze the effectiveness of the proposed method and perform extensive experiments to show that the proposed method not only provides a strict privacy guarantee, but also significantly improves the performance from the previous existing solutions. 
\end{abstract}

}

\maketitle

 \thispagestyle{fancy}
\lhead{This paper appears in IEEE Transactions on Services Computing. \url{https://doi.org/10.1109/TSC.2020.3039336} \\ Please feel free to contact us for questions or remarks.}
\rhead{\thepage}
\cfoot{}
\pagestyle{fancy}

\IEEEdisplaynontitleabstractindextext

%
\IEEEpeerreviewmaketitle

\IEEEraisesectionheading{\section{Introduction}\label{sec:introduction}}

%
%
%
%



\IEEEPARstart{M}{obile} crowdsourcing is an important platform in Internet of Things (IoT) paradigm for data collection from a number of sources such as sensors, mobile devices, and vehicles, and it serves as an important block for IoT applications, 
such as environment monitoring \citep{rana2010ear}, traffic condition detection \citep{mohan2008nericell}, 
and point-of-interest characterization \citep{chon2012automatically}. 
In the field of IoT applications, location remains to be the most successful and widely used information, which has provided boosts in various industrial applications and enabled analysis of statistical information, which is essential in improving marketing strategies.
For example, travel agencies can implement crowdsourcing to acquire travellers' travel paths to identify popular travel routes which may help them in seizing a larger market share.

For a traditional crowdsourcing system, the server interacts with the participants directly after receiving the tasks from the requester. 
However, offloading the data to the cloud introduces unforeseeable delay and heavy communication burden, 
especially when multiple interactions are needed. 
A better alternative solution is to take advantage of the nearby infrastructures or devices to process the data and send the processed data back to the cloud. 
This process is called edge computing, also known as fog computing \citep{singh2019fog}. 
Though edge computing enhances crowdsourcing, the collected data can also be used by the adversary to make sensitive inferences. 
The traditional solutions rely on the premise that the fog node is trusted and a series of privacy strategies are performed by the fog node on top of the user's real data. 
However, the fog node may not be trusted either. 
Therefore, the information needs to be protected locally by perturbing it before it leaves the user's devices.

Differential privacy is a provable privacy notation that has emerged as a \textit{de facto} standard for preserving privacy in a variety of areas. 
However it assumes that the data aggregator (fog node) is trusted, which is not always the case in the real world. 
To solve this problem, 
local differential privacy has been proposed, which perturbs the user's data locally, providing a much stronger privacy protection for the user. 
Several methods have been proposed to find the frequent items (e.g. frequently visited travel paths) under local differential privacy protection. 
However, the performance is not satisfactory, it either needs a big privacy budget or a large number of participants to ensure sufficiently high accuracy. 
This is because local differential privacy adds noise to each data record, which reduces the statistical accuracy significantly.

In this paper, we consider the location privacy problem for hot path statistics in the crowdsourcing system under fog computing architecture. 
To provide a much more accurate statistic while preserving the user's location information locally, 
we propose a novel solution that combines  both local differential privacy technology and secret sharing technology. 
The effort of combining the two techniques is not trivial. In particular, schemes constructed by sequentially performing the two techniques one after the other would not provide the desired increase of statistical accuracy. Furthermore, we also need to consider the significant increase of communication cost for users when any general secret sharing scheme is involved.
To solve these problems, we propose to perform randomize response to the sampling process and let users report their real value through secret sharing. 
Therefore, the user's location information is protected and the statistical error only comes from sampling process, which causes the significant improvement in the accuracy. 
On the other hand, due to the honest responses after the sampling process, the security assumption of our design is a more relaxed-yet-realistic assumption compared to the assumption used in local differentially private protocols. More specifically, we assume that the adversary can only corrupt some of the workers but not all of them. To show that the assumption is reasonable, we also provide the probability analysis of such assumption to be violated, i.e. the probability that the adversary manages be more powerful than assumed.

Overall, our main contributions are shown as follows.
\begin{itemize}
\item We propose a new crowdsourcing framework under fog computing architecture. 
Under the proposed framework, the tasks are partitioned and sent to each fog node, 
which not only improves the efficiency of the iterative statistics,  
but also reduces the overall communication cost to the cloud. 
\item We propose a novel solution to protect user's location information for hot travel path statistics. 
The proposed solution combines the secret sharing technology and local differential privacy,  
which provides a good balance between privacy, utility, and communication cost. 
\item We propose two secure reporting methods, Secure Reporitng (SR) and Enhanced Secure Reporting (ESR). 
SR enables the worker to report to a single node on the tree, which significantly reduces the communication cost. 
ESR requires the workers to report to multiple nodes, which enhances the privacy at the cost of small communication cost. 

\item We provide a theoretical analysis of the privacy, utility, and complexity of the proposed method. Besides, we did extensive experiments to evaluate the performance of the proposed methods over both real and synthetic datasets. The experimental results show that our methods perform much better than the state-of-the-art. 
\end{itemize}

The rest of the paper is organized as follows. In Section \ref{Pre}, we introduce the preliminaries. 
We propose our private hot path statistic methods and theoretically analyze their privacy and
utility in Sections \ref{HPS} and \ref{PAU}, respectively. Section \ref{EE} is dedicated to the discussion of the experimental result of the proposed methods. Section \ref{RW} discusses related
works and Section \ref{C} concludes the paper.

\section{Preliminaries} \label{Pre}

\subsection{Local differential privacy}

Different from traditional differential privacy \citep{dwork2006calibrating}, the Local differential privacy perturbs the user's data locally before it leaves the users' devices. 
Only the data owner can access the original data, which provides stronger privacy protection for the user. 
The formal definition is shown as follows:

\begin{defn}[Local Differential Privacy \citep{dwork2014algorithmic}]  

Let $\epsilon\geq 0$ and denote by $\mathcal{D}$ and $\mathcal{R}$ the set of possible data owned by a worker and his possible responses respectively. An algorithm $\mathcal{M}:\mathcal{D}\rightarrow\mathcal{R}$ satisfies $\epsilon$-local differential privacy if for any $t,t'\in \mathcal{D}$ and every possible output subset $S \subseteq \mathcal{R},$ we have
\begin{equation*}
Pr[\mathcal{M}(t)\in S] \leq e^\epsilon Pr[\mathcal{M}(t')\in S].
\end{equation*}
\end{defn}
Intuitively, this implies that given any possible response $S\in \mathcal{R},$ regardless of any background knowledge, the aggregator cannot identify the actual data owned by the worker with much confidence. 


\begin{theorem}[Sequential Composition \citep{dwork2014algorithmic}]
Suppose a method $\mathcal{M} =\{\mathcal{M}_1, \mathcal{M}_2, ..., \mathcal{M}_m \}$ has $m$ steps, 
which are sequentially performed on the same dataset. If each $\mathcal{M}_i$ provides $\epsilon_i$-differential privacy guarantee, then $\mathcal{M}$ provides $(\sum_{i=1}^m \epsilon_i)$-differential privacy. 
\end{theorem}
\begin{theorem}[Parallel Composition \citep{dwork2014algorithmic}]
Given a set a privacy mechanisms $\mathcal{M} =\{\mathcal{M}_1, \mathcal{M}_2, ..., \mathcal{M}_m \}$. 
If each $\mathcal{M}_i$ provides $\epsilon_i$-local differential privacy guarantee on a disjointed record of the entire dataset, $\mathcal{M}$ provides $\max\{\epsilon_1,\cdots,\epsilon_m\}$-local differential privacy. 
\end{theorem}

The composition theory states that if multiple  differentially private algorithms act on the same dataset sequentially, the total privacy level equals to the sum of the privacy budget of each differential privacy algorithm. If the algorithms act on the disjoint datasets, the total privacy level equals to the biggest privacy budget. 

\subsection{Randomized response}

Randomized response was proposed by Warner et al. \citep{warner1965randomized} 
as a survey technology to eliminate evasive answer bias. 
It is a typical mechanism to achieve local differential privacy. 
%

%
Let $R$ be a set of $d$ possible true values that a user can have and let $t\in R$ be the value of a user $w.$ Denote by $\hat{t},$ a random variable which represents the response of the user $w$ with sample space $R.$ 
The generalized randomized response works as follow, for any $v\in R$

\begin{equation}
Pr[\hat{t}=v]=\begin{cases}
p=\frac{e^\epsilon}{e^\epsilon+d-1},& \mathrm{if} ~t=v \mathrm{~and}\\
q=\frac{1}{e^\epsilon+d-1},& \mathrm{if} ~t\neq v
\end{cases}.
\end{equation} \label{PDRR}


The generalized randomized response outputs the true value with probability $\frac{e^\epsilon}{e^\epsilon+d-1}$ and outputs the one of other values with probability $\frac{1}{e^\epsilon+d-1}$.

 
\subsection{Additive secret sharing}

Secret sharing is a mechanism that distributes the data between participants without giving any of them the direct access to the original data (secret), while still enabling computations \citep{shamir1979share}. 
Let $x$ be a secret value to be shared between $g$ parties, we denote the shares of $x$ to be $\langle x \rangle = (x_0,\cdots, x_{g-1})$, such that $\sum_{i=0}^{g-1} x_i=x$ and party $i$ holds the share $x_i$. 
When $x$ is shared this way, we say that $x$ is additively shared.
Additive secret sharing is homomorphic with respect to addition, which enables local addition with low computation and communication complexity.

\section{Hot path statistic with location privacy-preserving} \label{HPS}

\subsection{Problem definition and system model}

We define the research problem in this section and present the framework of the mobile crowdsourcing system.

\subsubsection{Problem definition}

In this paper, we consider the private hot path statistic problem in a crowdsourcing system with a limited number of workers. Specifically, workers choose tasks released by the requester and, in turn, complete the tasks by reporting their real-time location information privately. 
The objective of the secure hot path statistic problem is to utilize the collected location information to identify some popular travel paths with high rate of visitations while preserving the privacy of the workers' real location information.
Table \ref{tab-parameter} list the notations used in this paper.

\begin{table}[htpb]  \centering
  \caption{Notations}\label{tab-parameter}
  \begin{tabular}{cl}
\hline
    Notation & Description   \\
\hline
$n$ &The number of workers\\
$m$ &The number of locations\\
$d_1,\cdots, d_{|D|}$  &The number of surviving node in each iteration\\
$ \varsigma $ &  A typical node in the Trie \\
$\mathbf{v}$  &  A worker's location vector   \\
$w$  & A typical worker   \\
    $D_1,\cdots, D_{|D|}$  & Date of the statistic   \\
    $L$  & The set of locations  \\
 $\Theta=(\theta_1,\cdots, \theta_{|D|})$ & Threshold for computation for each iteration\\
    $p_{\varsigma}$ & The prefix of the node $\varsigma$  \\
    $P$ & The prefix set\\
   $c(p_\varsigma)$ & The count number of the prefix $p_{\varsigma}$\\
    $\mathcal{P}_{\mathcal{H}}$ & The set of hot paths \\ 
\hline
\end{tabular}
\end{table}

\subsubsection{Framework}

We propose a crowdsourcing framework under fog computing architecture, 
which introduces a fog layer to relieve the cloud burden and improve efficiency. 
Fig. \ref{fig_frame} shows the proposed crowdsourcing framework. 
There are four components, which are shown as follows. 

\begin{figure}[!htp]
\centering
\includegraphics[width=3in]{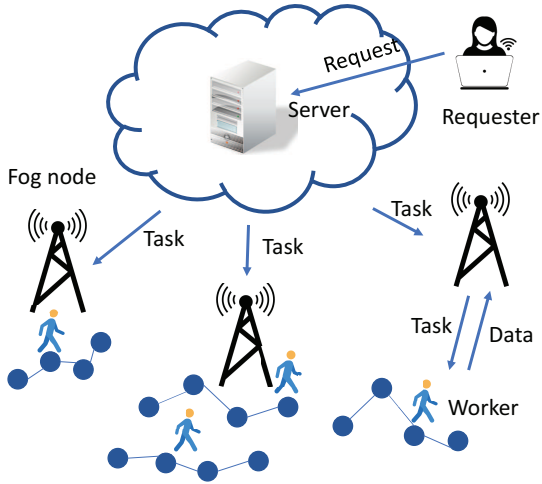}
\caption{The crowdsourcing framework}
\label{fig_frame}
\end{figure}

\begin{itemize}
\item \textbf{Requester:} The requester is the people or origination who would like to collect data for research or marketing purposes. They upload tasks to the crowdsourcing server and wait for the response. 

\item \textbf{Server:} For a traditional crowdsourcing system, the server is able to assign the tasks to the participants effectively, design the incentive mechanism, and manage the reputation values to ensure the high quality of the collected data, which are out of the scope of this paper. 
In the proposed framework, the server distributes the tasks to each fog node according to the region information and only need to aggregate the final statistic results from each fog node. 

\item \textbf{Fog nodes:} Each fog server is a highly virtualized computing system, similar to a light-weight cloud server, and is equipped with an on-board large-volume data storage. The fog node is in charge of assigning the tasks and collecting the data in the appointed region. There are two models of task assignment, Worker Selected Tasks (WST) and Server Assigned Tasks (SAT). In our framework, the participants select the tasks by themselves as we aim to collect the participants' travel paths within the appointed broad region. The participants can choose the task (fog node) according to their journeys. 

\item \textbf{Workers:} Workers select the task according to their journeys and report their location information to the specific fog node accordingly. 
To prevent their real location privacy from being disclosed while ensuring the accuracy of the statistic result, the worker would like to follow the privacy protocol and coordinate with each other. 

\end{itemize} 

\subsubsection{Threat model and security assumption}

In the proposed framework, the server and the fog nodes do not need to be trusted.
Furthermore,
we assume the existence of a semi-honest adversary $\mathcal{A}$ that may choose some users and fog nodes to
corrupt. 
When a user $w$ or a fog node $F$ is corrupted by $\mathcal{A}, \mathcal{A}$ has the access to all the data that is accessible to $w$ or $F$ respectively. The term semi-honest refers to the fact that the adversary $\mathcal{A}$ only has the power to access all the data of the corrupted parties. However, he does not have any power to manipulate the data to disrupt the calculation. The schemes that are discussed in this work can actually be improved to defend against active adversary that manipulates the calculation data. This can be achieved by using a similar verification techniques that are used in Mult-party computation (MPC) schemes such as SPDZ \citep{damgaard2012multiparty}. However, this work focuses on the semi-honest adversary since it is a more common assumption in the field of differential privacy.

\subsection{Hot path statistic LDPss}
To protect the workers' location information, we propose a trie-based private statistic method. The basic idea is that the fog node collects the statistic of different paths iteratively. For each iteration, we adaptively sample part of the workers and let them report their value through secret sharing. 
Specifically, the proposed method contains following three main components. 
\begin{itemize}
    \item \textit{Iterative statistics.} Iterative statistics computes the frequent paths iteratively and prunes the less frequent prefixes during each iteration. 
    \item \textit{Adaptive sampling.} Adaptive sampling privately samples more `functional' workers who can contribute to the final statistics to report their values. 
    \item \textit{Secure reporting.} Secure reporting enables workers to choose one or multiple nodes in the trie and report their values following secret sharing. 
\end{itemize}
The proposed method finds hot paths without going through all the possible combinations and get a more accurate statistics by performing the randomized response to sampling process instead of perturbing the data itself. The process of sampling includes sampling of workers and sampling of paths. The use of secret sharing technique in the reporting phase enables the participating workers to report their true value without disclosing it to the aggregator. By limiting the error to come only from the sampling process, we can expect a drastic increase in the statistical accuracy.

\subsubsection{Iterative statistics}
When the number of locations is large, the combination is innumerable, which means it is infeasible to query all the travel path. 
Tree construction method can solve this problem effectively \citep{wang2018privtrie}. 
Therefore, we propose a trie-based statistic method, which computes the frequent paths iteratively and prunes the less frequent prefixes during each iteration. 

Algorithm \ref{LDPss} shows the whole process of the iterative statistics. 
With all nodes being present initially, in order to protect the real proportion of workers that are functional in the second iteration, half of the workers are sampled to participate in selecting the root nodes at the beginning (Line $1$). 
Specifically, the selected workers encode their location as a location vector $\mathbf{v}=\mathbf{e}_i\in\mathbb{R}^m$, where $\mathbf{e}_i$ is the vector of length $m$ with all entries being zero except for the $i$-th position being $1$ to indicate the worker's true location (Line $3$). 
They report the location information through secret sharing (Line $4$-$7$). 
Having the count information for each node, nodes with insufficient count are then pruned while the remaining nodes serve as the root nodes for the tree construction (Line $9$-$11$).
For each of the subsequent dates, the tree is grown by adding children to the surviving prefixes (Line $14$-$15$). After \textit{Adaptive Sampling} and \textit{Secure Reporting} process are performed, the count of each prefix is then computed (Line $16$-$20$). Given the threshold $\theta_j,$ prefixes with insufficient count are then pruned (Line $21$-$25$). This process is repeated until the last date is reached.
In the end, the top-$k$ frequent prefixes are selected as the top-$k$ hot paths (Line $27$-$29$).   

\begin{figure}[!htp]
\centering
\includegraphics[width=3.5in]{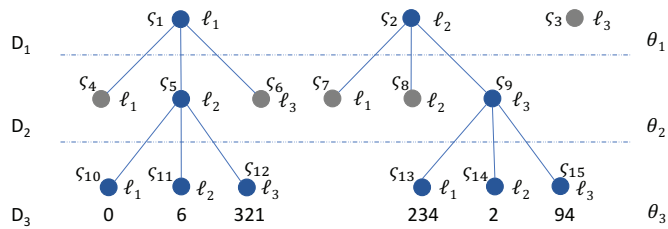}
\caption{Example of the statistic based on the trie}
\label{fig_tree}
\end{figure}

Fig. \ref{fig_tree} shows an example of the tree construction process. 
Initially, only nodes $\varsigma_1$ and $\varsigma_2$ survive while the node $\varsigma_3$ is pruned because its count number is smaller than $\theta_1$. 
Then locations $\ell_1, \ell_2, \ell_3$ are added to the surviving prefixes as their children. 
For the second iteration, only prefixes $(\ell_1,\ell_2)$ and $(\ell_2,\ell_3)$ survive. 
Therefore, the nodes $\varsigma_4$, $\varsigma_6$, $\varsigma_7$, and $\varsigma_8$ are pruned. 
In the end, we find that the path $(\ell_1,\ell_2,\ell_3)$ is the most popular travel path. 

\begin{algorithm}[!htbp]
\caption{LDPss}\label{LDPss}
\hspace*{0.02in} {\bf Input:} 
$D=\{D_1,\cdots, D_{|D|}\}$, location set $L=\{\ell_1, \ell_2, ..., \ell_d\}$, threshold $\Theta=\{\theta_1,\cdots,\theta_{|D|}\}$, requested number of hot paths $k$, secret sharing parameter $g$, privacy budget $\epsilon$, (If $ESR$ is used, report parameter $\alpha$)\\
\hspace*{0.02in} {\bf Output:}
top $k$ hot paths $\mathcal{P}_{\mathcal{H}}$
\begin{algorithmic}[1]
\STATE Randomly sample $n/2$ workers $w_i$, $i\in |\frac{n}{2}|$;

\FOR{worker $w_i=w_1,\cdots, w_{n/2}$ with location $\ell_i\in L$}
\STATE Encode $\ell_i$ as $\mathbf{v}_i=\mathbf{e}_j\in \mathbb{R}^m$;
\STATE Generate $g$ random vectors $\mathbf{s}^{i}_0,\cdots, \mathbf{s}^{i}_{g-1}$ such that $\mathbf{v}_i=\mathbf{s}^{i}_0+\cdots+\mathbf{s}^{i}_{g-1};$
\STATE Send $\mathbf{s}^{i}_t$ to $w_{(i+t-1)\pmod{n/2}+1}$ for $t=1,\cdots, g-1;$ 
\STATE After all communications have been done, compute $R_{w_i}=\sum_{t=0}^{g-1} \mathbf{s}^{((i-1-t)\pmod{n/2}+1)}_t;$ 
\STATE Send $R_{w_i}$ to the fog node;
\ENDFOR
\STATE Server computes $\left(c(\ell_1),\cdots, c(\ell_d)\right)=\sum_{i=1}^{n/2} R_{w_i};$
\STATE $L'\leftarrow L\setminus\{\ell: c(\ell)<\theta_1\}$;

\STATE Initialize the root nodes as the surviving locations : $P_1\leftarrow L'$;
\STATE Choose $\epsilon_2,\cdots, \epsilon_{|D|}\geq 0$ such that $\sum_{j=2}^{|D|}\epsilon_j=\epsilon$;
\FOR{$j=2~ \mathrm{~to~} ~|D|$}

\STATE Let $P_{j-1} =\{p^{j-1}_{\varsigma_1},\cdots, p^{j-1}_{\varsigma_{d_{j-1}}}\};$
\STATE Construct $P_j$ by concatenating paths in $P_{j-1}$ with all possible locations in $L:$
\begin{equation*}
P_j=\bigcup_{1\leq t\leq d_{j-1}, \ell_i\in L}\left\{\left(p^{j-1}_{\varsigma_t}\|\ell_i\right)\right\}
\end{equation*}
Suppose that $P_j=\{p^j_{\varsigma_1},\cdots,p^j_{\varsigma_{d_j}}\};$
\FOR{$i=1 \mathrm{~to~} n$}
\STATE $w_i$ with prefix $p_{w_i}$ runs $t_i\leftarrow AS(P_j,\mathbf{v}_i,\epsilon_s)$ where $t_i$ is either ``participate'' or ``not participate'';
\ENDFOR
\STATE Let the participating workers be $W_j=\{w_{i_1},\cdots, w_{i_{d_j}}\}\subseteq \{w_1,\cdots, w_n\};$
\STATE Fog nodes and $W_j$ jointly run $SR(P_j,g,\epsilon_s)$ or $ESR(P_j,\alpha,g,\epsilon_s,\epsilon_r)$ to get the estimate of $c(p^j_{\varsigma_i})$ corresponding to $p^j_{\varsigma_i}$, denoted as $\hat{\pi}$;
\FOR{$i=1 \mathrm{~to~} d_j$}
\IF{$\hat{\pi}(p_{\varsigma_i})<\theta_j$}
\STATE $P_j\leftarrow P_j\setminus\{p_{\varsigma_i}\};$
\ENDIF
\ENDFOR
\ENDFOR
\STATE $\mathcal{P}_{\mathcal{H}}=\varnothing$;
\STATE Sort $P_{|D|}$ according to the count number;
\STATE $\mathcal{P}_{\mathcal{H}} \leftarrow$ top $k$ prefixes in $P_{|D|}$;
\RETURN{$\mathcal{P}_{\mathcal{H}}$}
\end{algorithmic}
\end{algorithm}

\subsubsection{Adaptive sampling}

In the proposed method, only `functional' workers can contribute to the path statistic. Formally, given a set of node $\{\varsigma_1, \varsigma_2, ..., \varsigma_d\}$ in the trie, if a worker $w_i$ has location with prefix $p_{w_i}\in \{p_{\varsigma_1}, p_{\varsigma_2}, ..., p_{\varsigma_d}\}$, the worker is considered as a functional worker. The adaptive sampling process applies randomized response. 

Algorithm \ref{ASample} shows the detail of the sampling process. 
First, the worker identifies whether he is a functional worker in the first $5$ steps of the algorithm.
If the worker's location vector has a prefix that is one of the prefixes in $P$, then the variable $t$ is marked as $1$ and $0$ otherwise. 
For example, in Fig. \ref{fig_tree}, when the fog node collects the counts of node prefixes in $D_3$, 
if a worker's location has a prefix $(\ell_1,\ell_2,\ell_2)$, the worker is a functional worker. 
On the other hand, if a worker's location has prefix $(\ell_1,\ell_3,\ell_2)$, he is a non-functional worker that cannot contribute to the statistic.
The sampling process is depicted in steps $6$ through $9.$ Specifically, the worker uses randomized response to perturb his indicator $t$ to $\hat{t}.$ The worker is participating in the current iteration if $\hat{t}=1$ and he is not a participant otherwise.

\begin{algorithm}[htbp]
\caption{Adaptive Sampling ($AS(P, p_w,\epsilon)$)}\label{ASample}
\hspace*{0.02in} {\bf Input:}
Prefix set $P$, Worker $w$'s location prefix $p_{w},$ privacy budget $\epsilon$ \\
\hspace*{0.02in} {\bf Output:}
$\hat{t}\in\{0,1\}$ 
\begin{algorithmic}[1]
\IF{$p_w\in P$}
\STATE $t\leftarrow 1;$
\ELSE
\STATE $t\leftarrow 0;$
\ENDIF
\STATE Sample $\hat{t}\in\{0,1\}$ randomly such that
\begin{equation*}
\mathrm{Pr}(\hat{t}=1)=\left\{
\begin{array}{cc}
\frac{e^\epsilon}{e^\epsilon+1},&\mathrm{~if~} t=1\\
\frac{1}{e^\epsilon+1},&\mathrm{~if~} t=0\\
\end{array}
\right.
\end{equation*}

\end{algorithmic}
\end{algorithm}

Instead of letting the aggregator (fog node) to sample the worker, 
the proposed adaptive sampling happens in the worker side. 
The worker can adaptively choose whether to report his/her location according to the prefixes to be counted.  
The randomized sampling protects the workers' location information with local differential privacy guarantee. The fog node has no idea whether the reported node has the same prefix with the worker. 

\subsubsection{Secure reporting}

We propose two secure reporting methods, secure reporting (SR) and enhanced secure reporting (ESR), 
which select one or multiple nodes in each iteration and report the value through secret sharing. 
Specifically, if the workers $w_1,\cdots, w_{n_{\varsigma_i}}$ choose to report to a node $\varsigma_i,$ each worker additively shares his true value to other $g-1$ workers. Based on these shares, each worker can then generate an additive share of the sum which is then reported to the fog node. This ensures the accuracy of the sum while perfectly hiding each worker's true value from the fog node. \\ 
\\
\textbf{SR: Secure Reporting}
\begin{algorithm}[htbp]
\caption{Secure Reporting ($SR~(P, g,\epsilon))$}\label{SR}
\hspace*{0.02in} {\bf Input:}
 prefix set $P=\{p_{\varsigma_1},\cdots, p_{\varsigma_{d^\ast}}\}$, secret sharing scheme parameter $g,$ privacy budget $\epsilon$, participating workers $w_1,\cdots, w_{n^\ast}$\\
\hspace*{0.02in} {\bf Output:}
$\hat{\pi}$ 
\begin{algorithmic}[1]

\FOR{$i=1\mathrm{~to~}n^\ast$}
\STATE Let $w_i$ have location prefix $p_{w_i}$;
\IF{$p_{w_i}= p_{\varsigma_{j}}\in P$}
\STATE Report his intention to report to node $\varsigma_j$;
\ELSE
\STATE Choose any surviving node $\varsigma_j$ uniformly at random and report his intention to report to $\varsigma_j$;
\ENDIF
\ENDFOR
\FOR{$i=1\mathrm{~to~}d^\ast$}
\STATE Let $W_{\varsigma_i}=\{w_1,\cdots, w_{n_{\varsigma_i}}\}$ be the workers that has reported their intention to report to node $\varsigma_i$;
\IF{$n_{\varsigma_i}<g$}
\STATE  $c(p_{\varsigma_i})=0;$
\ELSE
\FOR{$w_j\in W_{\varsigma_i}$ with prefix $p_{w_j}$}
\IF {$p_{w_j}=p_{\varsigma_i}$}
\STATE $r_{w_j}=1$;
\ELSE 
\STATE $r_{w_j}=0$;
\ENDIF 
\STATE $R_{w_i}\leftarrow SS(w_i,r_i)$
\STATE Report $R_{w_j}$ to node $\varsigma_i$. 
\ENDFOR
\STATE fog node receives reports and estimate the number of the path for each node as  $\hat{\pi}(p_{\varsigma_i})=\frac{e^\epsilon+1}{e^\epsilon}c(p_{\varsigma_i})$ where $c(p_{\varsigma_i})=\sum_{j=1}^{n_{\varsigma_i}} R_{w_j}$; 
\ENDIF
\ENDFOR

\end{algorithmic}
\end{algorithm}

\begin{algorithm}[htbp]
\caption{Secret Sharing ($SS(w_i,r_i)$)}\label{ss}
\hspace*{0.02in} {\bf Input:} 
worker $w_i$, value $r_i$\\
\hspace*{0.02in} {\bf Output:}
report $R_{w_i}$
\begin{algorithmic}[1]

\STATE Generate $g$ random vectors $\mathbf{s}^{i}_0,\cdots, \mathbf{s}^{i}_{g-1}$ such that $r_i=\mathbf{s}^{i}_0+\cdots+\mathbf{s}^{i}_{g-1};$
\STATE Send $\mathbf{s}^{i}_t$ to $w_{(i+t-1)\pmod{n/2}+1}$ for $t=1,\cdots, g-1;$ 
\STATE After all communications have been done, compute $R_{w_i}=\sum_{t=0}^{g-1} \mathbf{s}^{((i-1-t)\pmod{n/2}+1)}_t;$ 
\RETURN{$R_{w_i}$}
\end{algorithmic}
\end{algorithm}
Given a functional worker $w$, benefiting from the Adaptive Sampling, 
the worker can only report a single node in the trie without disclosing his true location information. 
Specifically, if $p_w=p_{\varsigma_i}$, the worker only needs to report to the node $\varsigma_i$. If $w$ is a non-functional worker, he chooses a surviving node $\varsigma_j$ randomly and report to it. 
To effectively broadcast the shares, the contributing workers need to report their intention so each contributing worker knows the set of nodes that will report to the same node as him. 
As shown in Algorithm \ref{SR}, workers report their intentions accordingly (Line $1$-$8$). 
Once the intention reporting has been done, the calculation through secret sharing begins. 
Firstly, for nodes with less than $g$ participating workers, the computation phase is omitted and the number of path is recorded as $0$ (Line $10$-$12$). 
This is because there are not enough workers to conduct the protocol and the count number must be smaller than $g$. So if $g$ is set to be sufficiently small such that $\theta_j>g$ for any $j,$ this node will be pruned anyway.
For nodes with at least $g$ participating workers, the calculation phase is performed. 
For the calculation of a node $\varsigma$ with corresponding prefix $p_\varsigma,$ each participating worker $w_j$ holds a private value $r_{w_j}$ which has value $1$ if $p_{w_j}=p_{\varsigma}$ and $0$ otherwise (Line $14$-$19$). 
The participating workers perform a secret sharing protocol described in Algorithm.\ref{ss} to report the sum of their private values to the fog node(Line $20$-$21$). 
The fog node can estimate the number of the path by calculating
\begin{equation}
 \hat{\pi}(p_{\varsigma_i})=\frac{e^\epsilon+1}{e^\epsilon}c(p_{\varsigma_i}), 
\end{equation}
where $c(p_{\varsigma_i})=\sum_{j=1}^{n_{\varsigma_i}} R_{w_j}$.

\begin{figure}[!htp]
\centering
\includegraphics[width=3in]{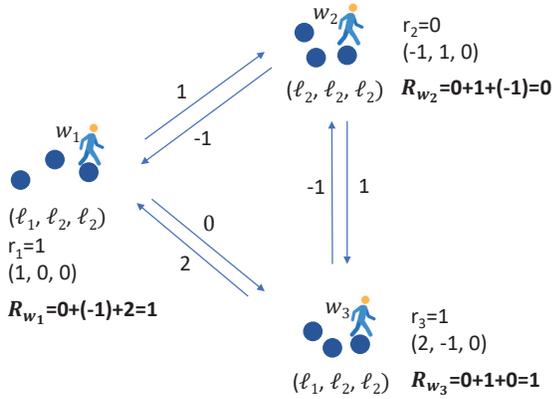}
\caption{Example of the secret sharing}
\label{fig_reportexample}
\end{figure}

Fig. \ref{fig_reportexample} shows an example of the secret sharing process. 
Assume the fog node collects the count of node $\varsigma_{11}$, workers $w_1, w_2$ and $w_3$ intend to report to this node. Workers $w_1$ and $w_3$ have the same prefix as $p_{\varsigma_{11}}$, so their private value is set to be $1$. On the other hand, $w_2$ is  a non-functional worker, who also intends to report to the node $\varsigma_{11}$. 
Therefore, $w_2$ sets his secret value $r_2=0$. Each of them generates $3$ shares and send two shares to other workers. 
In the end, each of them reports $R_{w_i}$ to the node. 
And we find that $\sum R_{w_i}=2$, which is the true value of the statistic.

Though workers report to a specific node, the fog node cannot distinguish whether a worker that reports to it is functional. 
However, though the fog node cannot infer the worker's real location information, 
the fog node can eliminate all the other surviving nodes from being the possible location information of a reporting worker.
Therefore, the worker's location in Algorithm \ref{SR} is actually hidden among the reported node and pruned nodes. 
That is, the worker's location information during the reporting process is protected by $K$-anonymity where $K-1$ is the number of pruned nodes.
Hence, to ensure privacy, we need to ensure some nodes to be pruned when selecting the root nodes.
And, a bigger threshold at the beginning contributes to a bigger $K$. \\
\textbf{ESR: Enhanced Secure Reporting}

To provide more rigorous privacy protection, we further propose a enhanced secure reporting method. 
Instead of reporting to a single node, the workers report to multiple nodes to enhance the privacy protection. 
Under ESR method, the fog node knows nothing except the noisy statistics. 
Worker's location is strictly protected by differential privacy technology.


\begin{algorithm}[htbp]
\caption{Enhanced Secure Reporting ($ESR ~ (P, \alpha, g,\epsilon_s,\epsilon_r))$}\label{ESR}
\hspace*{0.02in} {\bf Input:}
 prefix set $P=\{p_{\varsigma_1},\cdots,p_{\varsigma_{d^\ast}}\}$, secret sharing parameter $g,$ report parameter $\alpha,$ privacy budgets $\epsilon_s, \epsilon_r$, participating workers $w_1,\cdots, w_{n^\ast}$\\
\hspace*{0.02in} {\bf Output:}
$\hat{\pi}$ 
\begin{algorithmic}[1]
\STATE Let $w_i$ have location prefix $p_{w_i}$;
\STATE $\mathcal{P}_I=\{P'\subseteq P: p_{w_i}\in P',|P'|=\alpha d^\ast\},$;
\STATE $\mathcal{P}_E=\{P'\subseteq P: p_{w_i}\not\in P',|P'|=\alpha d^\ast\}$;
\STATE $\mathcal{P}=\{P'\subseteq P: |P'|=\alpha d^\ast\}$;
\FOR{$i=1$ to $n^\ast$}
\IF{$p_{w_i}=p_{\varsigma_j}\in P$}
\STATE Sample $\mathcal{P}_{w_i}\in \mathcal{P}$ such that for any $P^\ast\in \mathcal{P}$, the probability that $\mathcal{P}_{w_i}=P^\ast$ is defined to be $\mathrm{pr}$ where
\begin{equation*}
\mathrm{pr}=\left\{
\begin{array}{cc}
\frac{e^{\epsilon_r}}{\binom{d^\ast-1}{\alpha d^\ast -1}\cdot e^{\epsilon_r} + \binom{d^\ast-1}{\alpha d^\ast}} &\mathrm{~if~} P^\ast\in \mathcal{P}_I\\
\frac{1}{\binom{d^\ast-1}{\alpha d^\ast -1}\cdot e^{\epsilon_r} + \binom{d^\ast-1}{\alpha d^\ast}} &\mathrm{~if~} P^\ast\in \mathcal{P}_E\\
\end{array}
\right.;
\end{equation*}
\ELSE
\STATE Sample $\mathcal{P}_{w_i}\in \mathcal{P}$ such that for any $P^\ast\in \mathcal{P}$, 

\begin{equation*}
\mathrm{pr}[\mathcal{P}_{w_i}=P^\ast]=
\begin{array}{cc}
\frac{1}{\binom{d^\ast}{\alpha d^\ast}};
\end{array}
\end{equation*}

\ENDIF
\STATE $w_i$ reports his intention to report to all nodes $\varsigma_j\in \mathcal{P}_{w_i}$ 
\ENDFOR

\FOR{$i=1\mathrm{~to~}d^\ast$}
\STATE Let $W_{\varsigma_i}=\{w_1,\cdots, w_{n_{\varsigma_i}}\}$ be the workers that have reported their intention to report to node $\varsigma_i$;
\IF{$n_{\varsigma_i}<g$}
\STATE $c(p_{\varsigma_i})=0;$
\ELSE

\FOR{$w_j\in W_{\varsigma_i}$ with prefix $p_{w_j}$}
\IF {$p_{w_j}=p_{\varsigma_i}$}
\STATE $r_{w_j}=1$;
\ELSE 
\STATE $r_{w_j}=0$;
\ENDIF 
\STATE $R_{w_j}\leftarrow SS(w_j,r_j)$
\STATE Report $R_{w_j}$ to node $\varsigma_i$; 
\ENDFOR
\STATE fog node receives reports and estimate the number of the path for each node as $\hat{\pi}(p_{\varsigma_i})=\frac{(e^{\epsilon_s}+1)\left(e^{\epsilon_r}+\frac{1-\alpha}{\alpha}\right)}{e^{\epsilon_s+\epsilon_r}}c(p_{\varsigma_i})$ where $c(p_{\varsigma_i})=\sum_{j=1}^{n_{\varsigma_i}} R_{w_j}$;

\ENDIF
\ENDFOR

\end{algorithmic}
\end{algorithm}

As shown in Algorithm \ref{ESR}, 
suppose that after the first phase of the iteration, a worker $w_j$ with location prefix $p_{w_j}$ is selected to be participating (Line 1). 
Let $\mathcal{P}_I$ be the prefix set that includes $p_{w_j}$, 
$\mathcal{P}_E$ be the prefix set exclude $p_{w_j}$, 
and  $\mathcal{P}$ be all the possible prefix set with length $\alpha d^\ast$ (Line $2$-$4$). 
If $p_{w_j}\in P$, the worker chooses a subset $\mathcal{P}_{w_j}$ of $P$ of size $\alpha d^\ast$ such that any $\mathcal{P}_{w_j}\in \mathcal{P}_I$ is chosen with probability $\frac{e^{\epsilon_r}}{\binom{d^\ast-1}{\alpha d^\ast -1}\cdot e^{\epsilon_r} + \binom{d^\ast-1}{\alpha d^\ast}}$ and any $\mathcal{P}_{w_j}\in \mathcal{P}_E$ is chosen with probability $\frac{1}{\binom{d^\ast-1}{\alpha d^\ast -1}\cdot e^{\epsilon_r} + \binom{d^\ast-1}{\alpha d^\ast}}$ (Line $6$-$7$). 
On the other hand, if $p_{w_j}\not\in P$, then he chooses a subset of $P$ of size $\alpha d^\ast$ uniformly at random (Line $8$-$10$). 
Workers report their intentions to the fog node after making the private sampling (Line $11$). 
Then the worker's real location information is reported through additive secure reporting (Line $15$-$26$). 
The fog node can estimate the real count of each path by calculating
\begin{equation}
 \hat{\pi}(p_{\varsigma_i})=\frac{(e^{\epsilon_s}+1)\left(e^{\epsilon_r}+\frac{1-\alpha}{\alpha}\right)}{e^{\epsilon_s+\epsilon_r}}c(p_{\varsigma_i}), 
\end{equation}
where $c(p_{\varsigma_i})=\sum_{j=1}^{n_{\varsigma_i}} R_{w_j}$. 
Therefore, the protection of worker's location privacy is hence enhanced via further subset selection of the nodes to report on the tree.

%
%

In the remainder of this paper, we let \textbf{LDPss1} denote the proposed method that the worker reports to a single node via Algorithm~\ref{SR} and \textbf{LDPss2} denote the method that the worker reports to multiple nodes via Algorithm~\ref{ESR}.

\subsubsection{Discussion}

The proposed method achieves two targets that the traditional local differential privacy cannot achieve. 

First, the proposed method can achieve a much higher accuracy even if the privacy budget is quite small. 
The traditional local differential privacy generally has a big error for a much smaller privacy budget due to the big statistical variance. 
On the other hand, the error of the proposed method only comes from the sampling processes (e.g. worker selection and node selection), 
which are controlled by privacy budget $\epsilon$ and parameter $\alpha$ respectively. 
Even when the privacy budget is very small, we can still expect to have around half of the participating workers to contribute to the statistical count by adjusting the parameter $\alpha$. 
The sampling error is quite limited compared with adding noise to the data record directly. 

Second, the proposed method can achieve a good performance even when a limited number of workers participate.
According to the existing deployment, LDP needs millions of participants to ensure the statistical accuracy \citep{erlingsson2014rappor, adp, ding2017collecting}. 
Much more participants are needed for dataset with higher data dimension. 
However, the accuracy of the proposed method is not affected by the number of locations (dimension). 
This is because all the selected workers report the true value through secret sharing instead of a perturbed value to the fog node. 
As we mentioned, the error only comes from the sampling processes, which are independent of the data dimension and have a much smaller statistical variance. 
Therefore, the proposed method can work well for situations with a limited number of participants and is not affected by the data dimension.

\section{Privacy and utility analysis} \label{PAU}
 We provide a detailed analysis of the proposed method in this section.

\begin{lem} \label{sDP}
Let $\epsilon\geq 0$ be the privacy budget given as the input for Algorithm~\ref{ASample}. Then Algorithm \ref{ASample} satisfies $\epsilon$-local differential privacy. 
\end{lem}
\begin{proof}
For each level, the worker decides whether to report their value according to the fact that whether he/she is a functional worker.  
Let $P$ be the set of prefixes to be considered, $p_{w}$ be a worker's prefix of the travel path, and $\epsilon$ be the privacy budget assigned to each iteration, then for each iteration we have \begin{equation}
\frac{\mathrm{Pr}[AS(P,p_{w},\epsilon)=r]}{\mathrm{Pr}[AS(P,p_{w}',\epsilon)=r]}\leq\frac{e^{\epsilon}}{e^{\epsilon}+1}\left/\frac{1}{e^{\epsilon}+1}\right.=e^{\epsilon}
\end{equation}

Therefore, Algorithm \ref{ASample} satisfies $\epsilon$-local differential privacy. 

\end{proof}

In LDPss1, the Algorithm \ref{ASample} is performed in each iteration. 
According to the composition theory, if each iteration assigns $\epsilon_i$ privacy budget to Algorithm \ref{ASample}, the overall sampling process satisfies $\epsilon$-differential privacy, where $\epsilon=\sum\epsilon_i$. 

Compared to the traditional random sampling, 
the proposed adaptive sampling not only provide a straight differential privacy guarantee, but also enhance the accuracy of the statistic. The rationale is that the differential privacy sampling samples useful data with a higher probability, which results in a higher proportion of effective sample, while the traditional randomized sampling samples all data with the same probability, which results in less effective samples. 
The detailed analysis can be found in Supplementary material.

\begin{lem}\label{lem1}
For any $\epsilon > 0$ and $0<\alpha\leq 1$, we have LDPss2 provides $\epsilon$-local differential privacy. 
\end{lem}

\begin{proof} 
By the composition theorem of differential privacy, we need to prove each iteration of the LDPss2 satisfy $\epsilon$-local differential privacy, where $\sum_{i=2}^{|D|}\epsilon_i=\epsilon$.  

Recall that each iteration of LDPss2 consists of two perturbation sub-protocols which are detailed in Algorithms~\ref{ASample} and \ref{ESR}. As discussed in Algorithm~\ref{LDPss}, set $\epsilon_s,\epsilon_r\geq 0$ such that $\epsilon_i=\epsilon_s+\epsilon_r.$ Utilizing the same proof technique as used in Lemma~\ref{lem1}, it can be shown that the first half of the iteration provides $\epsilon_s$-local differential privacy. It remains to prove that the second half of the iteration, which is depicted in Algorithm~\ref{ESR} provides $\epsilon_r$-local differential privacy. 

Note that Algorithm~\ref{ESR} can be further divided into two phases, sampling phase and secret sharing phase. The secret sharing phase, which is done from Line 10 onwards, is statistically secure by Lemma~\ref{sssec}. Hence we only needs to ensure that the sampling phase which is done up to Line $9$ provides $\epsilon_r$-local differential privacy. Denote by $\mathcal{M}(p_{w})$ the sampling mechanism which is done for each worker $w$ in Lines $3$ up to $7$ in Algorithm~\ref{ESR} to output the set $\mathcal{P}_w\in \mathcal{P}$ of prefixes corresponding to the fog node that $w$ intends to report to. Then for a worker $w$ with prefix $p_w$,  

\begin{equation} \label{eqsr2}
\small\mathrm{Pr}[\mathcal{M}(p_w)=\mathcal{P}_w]=\left\{
\begin{array}{rl}
\frac{1}{\binom{d}{\alpha d}} & \mathrm{if~} p_w\not\in P\\
\frac{1}{\binom{d-1}{\alpha d-1}\cdot e^{\epsilon_2}  + \binom{d-1}{\alpha d}} & \mathrm{if~} p_w\in P \setminus \mathcal{P}_w\\
\frac{e^{\epsilon_2}}{\binom{d-1}{\alpha d-1}\cdot e^{\epsilon_2} + \binom{d-1}{\alpha d}}& \mathrm{if~} p_w\in \mathcal{P}_w.
\end{array}
\right.
\end{equation}
Denote the three probabilities to be $p_1,p_2$ and $p_3$ respectively. Since $\epsilon_r>0,$ we have $p_2\leq p_3.$ So it is sufficient to prove that $\max\left(\frac{p_1}{p_2},\frac{p_3}{p_1},\frac{p_3}{p_2}\right)\leq e^{\epsilon_r}.$
By Equation~\eqref{eqsr2},
\begin{equation*}
\frac{p_1}{p_2}=\frac{\binom{d-1}{\alpha d-1}\cdot e^{\epsilon_r}  + \binom{d-1}{\alpha d}}{\binom{d}{\alpha d}}=\alpha (e^{\epsilon_r}-1)+1\leq e^{\epsilon_r}.
\end{equation*}
Next, considering $p_3/p_1,$ since $\alpha>0, \epsilon_r>0,$ we have
\begin{equation*}
\frac{p_3}{p_1}=\frac{e^{\epsilon_r}\cdot \binom{d}{\alpha d}}{\binom{d-1}{\alpha d-1}\cdot e^{\epsilon_r}  + \binom{d-1}{\alpha d}}=\frac{e^{\epsilon_r}}{\alpha (e^{\epsilon_r}-1)+1}<e^{\epsilon_r}.
\end{equation*}

Lastly, it is easy to see that $\frac{p_3}{p_2}=e^{\epsilon_r}.$ This shows that for any possible locations $p_w,p_w'$ and any possible output $\mathcal{P}_w\in \mathcal{P}:$

\begin{equation*}
\frac{Pr[\mathcal{M}(p_w)=\mathcal{P}_w]}{Pr[\mathcal{M}(p_w')=\mathcal{P}_w]}\leq e^{\epsilon_r}.
\end{equation*}

Therefore, each iteration of LDPss2 satisfies $\epsilon_i$-local differential privacy and further LDPss2 satisfies $\epsilon$-local differential privacy. 

\end{proof}

The next lemma provides the accuracy analysis of the count estimate of each node in the trie.

\begin{lem}
For any node $\varsigma$ in the trie, let $\tilde{\pi}_1$ and $\tilde{\pi}_2$ be the estimated count the corresponding fog node stores after the execution of Algorithms~\ref{SR} and Algorithms~\ref{ESR} respectively. Then $\tilde{\pi}_1$ and $\tilde{\pi}_2$ are unbiased estimators of $\pi,$ i.e. $E(\tilde{\pi}_1)=E(\tilde{\pi}_2)=\pi$ and their variances are independent of the dimension of the value.
\end{lem}
\begin{proof} We consider the two methods separately. 
\\
\textit{LDPss1:}

Recall that in LDPss1, there is one sampling through Algorithm~\ref{ASample} to decide whether a worker participates in the calculation. Once this is decided, any participating worker only reports to one node truthfully. In other words, fixing a node $\varsigma$ in $i$-th depth of the trie, its estimate $\hat{\pi}_1$ only depends on the number of workers with the same prefix as $p_{\varsigma}$ that participate on the computation. By design, each of these workers has an identical and independent probability to participate in the computation of the count of node $\varsigma.$ Hence $\hat{\pi}_1= \frac{e^{\epsilon_i}+1}{e^{\epsilon_i}}c(\varsigma)$, where $c(\varsigma)$ can be seen as a variable that follows Binomial distribution with $\pi$ experiments and success probability $\frac{e^{\epsilon_i}}{e^{\epsilon_i}+1}$. 
Therefore, we have
\begin{align}
E(\hat{\pi}_1)&=\frac{e^{\epsilon_i}+1}{e^{\epsilon_i}} E[c(\varsigma)] 
=\frac{e^{\epsilon_i}+1}{e^{\epsilon_i}} \pi \frac{e^{\epsilon_i}}{e^{\epsilon_i}+1} 
=\pi
\end{align}
and 
\begin{align}
\nonumber\mathrm{Var}[\tilde{\pi}_1]&=(\frac{e^{\epsilon_i}+1}{e^{\epsilon_i}})^2 \mathrm{Var}[c(\varsigma)] \\
\nonumber&=\left(\frac{e^{\epsilon_i}+1}{e^{\epsilon_i}}\right)^2\pi \frac{e^{\epsilon_i}}{e^{\epsilon_i}+1} \frac{1}{e^{\epsilon_i}+1} \\
&=\frac{\pi}{e^{\epsilon_i}}
\end{align}


\textit{LDPss2:}

Similar to method LDPss1, 
the count of each node $c(\varsigma)$ is a binomial random variable with $\pi$ experiments. 
The success probability $p$ equal to the probability that a worker with location prefix $p_w$ reports to the node $\varsigma$ with $p_{\varsigma}= p_w$. Specifically, 
\begin{equation}
p=\frac{e^{\epsilon_s}}{e^{\epsilon_s}+1}\cdot\frac{e^{\epsilon_r}}{e^{\epsilon_r}+\frac{1-\alpha}{\alpha}}
\end{equation}

Recall that the estimation of the count $\hat{\pi}_2=\frac{1}{p}c(\varsigma)$, 
the expectation of the estimation is 
\begin{align}
E[\hat{\pi}_2]=\frac{1}{p}E[c(\varsigma)]=\frac{1}{p} \pi p = \pi 
\end{align}
and 
\begin{align}
\nonumber \mathrm{Var}[\hat{\pi}_2]&=\frac{1}{p^2}\pi p (1-p)\\
&=\pi\cdot\frac{e^{\epsilon_r}+\frac{1-\alpha}{\alpha}e^{\epsilon_s}+\frac{1-\alpha}{\alpha}}{e^{\epsilon_s+\epsilon_r}}
\end{align}

We find that except the inherent variance of the dataset, 
the variance is only affected by the sampling probability for both methods. 
Therefore, the high dimension dataset cannot affect the accuracy of the estimation. 
\end{proof}

\textbf{Computation cost.} The proposed method utilizes the additive secret sharing technology. 
Different from the complex encryption calculation, 
the proposed method only needs to perform several addition and sampling operations, which do not consume too much computational resources. 
Specifically, in each iteration, LDPss1 requires each worker to perform $g$ sampling process ($1$ to determine his participation to the calculation and $g-1$ to generate the secret shares of his true value) along with $g-1$ addition operations. 
For the second scheme, LDPss2, each worker is required to perform $(1+(g-1)\alpha d)$ sampling processes and $(g-1)\alpha d$ addition operations. 
Therefore, assuming that $g$ is a constant, the computation costs for the $i$-th iteration of LDPss1 and LDPss2 are $O(1)$ addition + $O(1)$ sampling and $O(\alpha d_i)$ addition + $O(\alpha d_i)$ sampling respectively.

\textbf{Communication cost.} In each iteration, the worker needs to report their intention, broadcast the random value to other $g-1$ workers, then reports the additive value to the fog node. For LDPss1, the worker only reports to a single node, the communication cost is $g+1$ values. 
For LDPss2, the worker reports to $\alpha d_i$ nodes. 
Therefore, the communication cost for LDPss2 is $\alpha d_i (g+1)$ values. 
It can be shown that the system provides information theoretical security against passive adversary. 
Therefore, assuming that $g$ is a constant, the communication costs for the $i$-th iteration of LDPss1 and LDPss2 are $O(1)$ and $O(\alpha d_i)$ values respectively.

\textbf{Efficiency analysis.} The proposed method is performed under fog computing architecture. 
Benefiting from the framework, the worker only needs to communicate with the local fog devices for every iteration. 
The fog node will send the final statistical result to the cloud. 
Therefore, instead of multiple rounds of long-distance communications between a large number of workers and cloud, under the proposed framework, only one round of communication is needed between each fog node and the cloud. All other communication happened locally between workers and the fog node. 
Compared to schemes that are not based on fog node architecture, this provides our scheme with a significant improvement in terms of the communication efficiency and latency.

\begin{remark} The significant advantage of the proposed method is that it combines the advantages of both encryption calculation and differential privacy. Table~\ref{compare} shows the comparison of these methods.  
\end{remark}

\newcommand{\tabincell}[2]{\begin{tabular}{@{}#1@{}}#2\end{tabular}}
\begin{table*}[!tph]  
\renewcommand\arraystretch{1.5}
\centering
\begin{threeparttable}[b]
  \caption{Comparison of four methods}\label{compare}
  \begin{tabular}{|c|c|c|c|c|c|}
\hline
    Methods & \tabincell{l}{Privacy} & Variance & \tabincell{l}{Computation Cost \\ (Addition + Sampling)}& \tabincell{l}{Communication\\ Cost} & Accuracy  \\
\hline

   SS \tnote{1}  & SS & $0$ & $ O(nd)+O(nd)$  & $ O(nd)$ & Accurate  \\
   \hline
    LDP  & $\epsilon$-LDP  & $\pi \frac{e^\epsilon(d-1)}{(e^\epsilon-1)^2}+(n-\pi)\frac{e^\epsilon+d-2}{(e^\epsilon-1)^2}  $ & $0+O(1)$ & $O(1)$ & Low  \\
      \hline
    LDPss1 &\tabincell{l}{$\epsilon$-LDP, SS, and \\K-anonymity} & $\frac{\pi}{e^\epsilon}$ &$O(1)+O(1)$ &$O(1)$ & High \\
    \hline
    LDPss2 & $\epsilon$-LDP and SS  & $\pi\cdot\frac{e^{\epsilon_2}+\frac{1-\alpha}{\alpha}e^{\epsilon_1}+\frac{1-\alpha}{\alpha}}{e^{\epsilon_1+\epsilon_2}}$ &$O(\alpha d)+O(\alpha d)$  &$O(\alpha d)$ & High \\
\hline
\end{tabular}
\begin{tablenotes}
  \item[1] SS: Secret Sharing
\end{tablenotes}
\end{threeparttable}
\end{table*}


The figures provided in Table~\ref{compare} is calculated under the assumption that all methods are used in a trie-based structure and the cost is calculated for a single user for each iteration of the protocol. We note that for all the methods that are considered, relative to the cost incurred in encryption based protocols, their costs are very small with some subtle differences between them. A protocol that is based on secret sharing scheme provides high utility as well as information theoretical security guarantee for all the data except for its output's privacy. In particular, it does not provide any differential privacy and linkage attack can be easily done. The high accuracy guarantee comes with a very high communication cost. On the other extreme, protocols based on traditional LDP techniques have a very small communication cost but its accuracy is limited especially for the case where there are limited number of participants. Our schemes provide some trade-off between the two traditional techniques. LDPss1 provides a high statistical accuracy with very small communication cost. However, this comes with the cost of a weaker privacy guarantee. LDPss2 provides a better privacy guarantee compared to LDPss1 with a slightly higher communication cost.

\begin{remark}
The proposed method has such a good performance is because the local differential privacy is applied on the sampling process to disturb the statistic instead of the data itself. 
Therefore, the privacy of the proposed method is not as strict as pure local differential privacy. 
However, the worker's data still can be protected well under the the proposed method own to the additive secret sharing. Lemma~\ref{sssec} shows the security analysis.
\end{remark}

\begin{lem}\label{sssec}
Fix a date $D_i\in D$ and a surviving node in date $D_i, \varsigma.$ Suppose that there are $n_\varsigma$ participating workers in the computation of $c(p_\varsigma).$ Then with high probability, the scheme provides statistical security against passive adversary that controls either
\begin{itemize}
\item up to $n_\varsigma-1$ participating workers or
\item the fog node along with at most $\min(2g-3,n_\varsigma-1)$ participating workers. 
\end{itemize}
\end{lem}

\begin{proof}
The proof details are given in the Supplementary material.
\end{proof}

Lemma~\ref{sssec} is based on the assumption that the participating workers for a node are fully controlled by the adversary, 
which is almost impossible to be achieved due to Adaptive Sampling. 
Therefore, the fog node needs to corrupt more workers to ensure enough corrupted workers report to the targeted node. A more detailed discussion on this can be found in the Supplementary material.
Lemma~\ref{sssec} provides a preliminary mean to determine the value of $g$ to balance between the security level of the scheme and its communication cost.

\section{Experimental evaluation}  \label{EE}

In this section, we evaluate the performance of the proposed hot travel path statistic method through an extensive set of experiments, and in particular, 
we show the impact of variable $k$, privacy budget $\epsilon$, number of participants, the number of locations and parameter $\alpha$ on the performance of the proposed method on both real and synthetic datasets. Besides, we compare the proposed method with state of the art works in terms of accuracy. 

\textbf{Datasets.} We tested the proposed method on two real datasets and one synthetic dataset. 
\begin{itemize}
\item \textit{Gowalla}. Gowalla data contains $6, 442, 8904$ check-in locations of $196,586$ users over the period of Feb. 2009 to Oct. 2010. We extracted all the check-ins in the range $[45, 55]\times [-5, 5]$ on the map and partitioned it into cells of $3 \times 3$. We assign each cell a unit local ID and match the user's check-in data with the corresponding location cell. 
\item \textit{BrightKite}. BrightKite data contains a total of $4,491,143$ check-ins of $58,228$ users over the period of Apr. 2008 - Oct. 2010. We extract the check-ins in New York, USA and partition the whole area into $4 \times 4$ cells. Same with Gowalla, each cell is assigned a unit location ID and the users' check-in information is mapped to the corresponding local cell. 
\item \textit{Synthetic}. We generate a synthetic dataset with $5000$ records. 
Specifically, each record is a location sequence with a length of $5$, which denotes the worker's travel path.  Furthermore, each value of the record can be any number from $1$ to $5$. That is, there are $5^5$ combinations. We sample the record from a Zipf distribution.  
\end{itemize}

Fig. \ref{FIG-frequency} shows the frequency of top-$30$ records ($5$ days travel path) of each dataset.

\begin{figure}[htbp]
\centering

\subfloat[Gowalla]{
\label{Fig-Gowalla}
\includegraphics[scale=0.22]{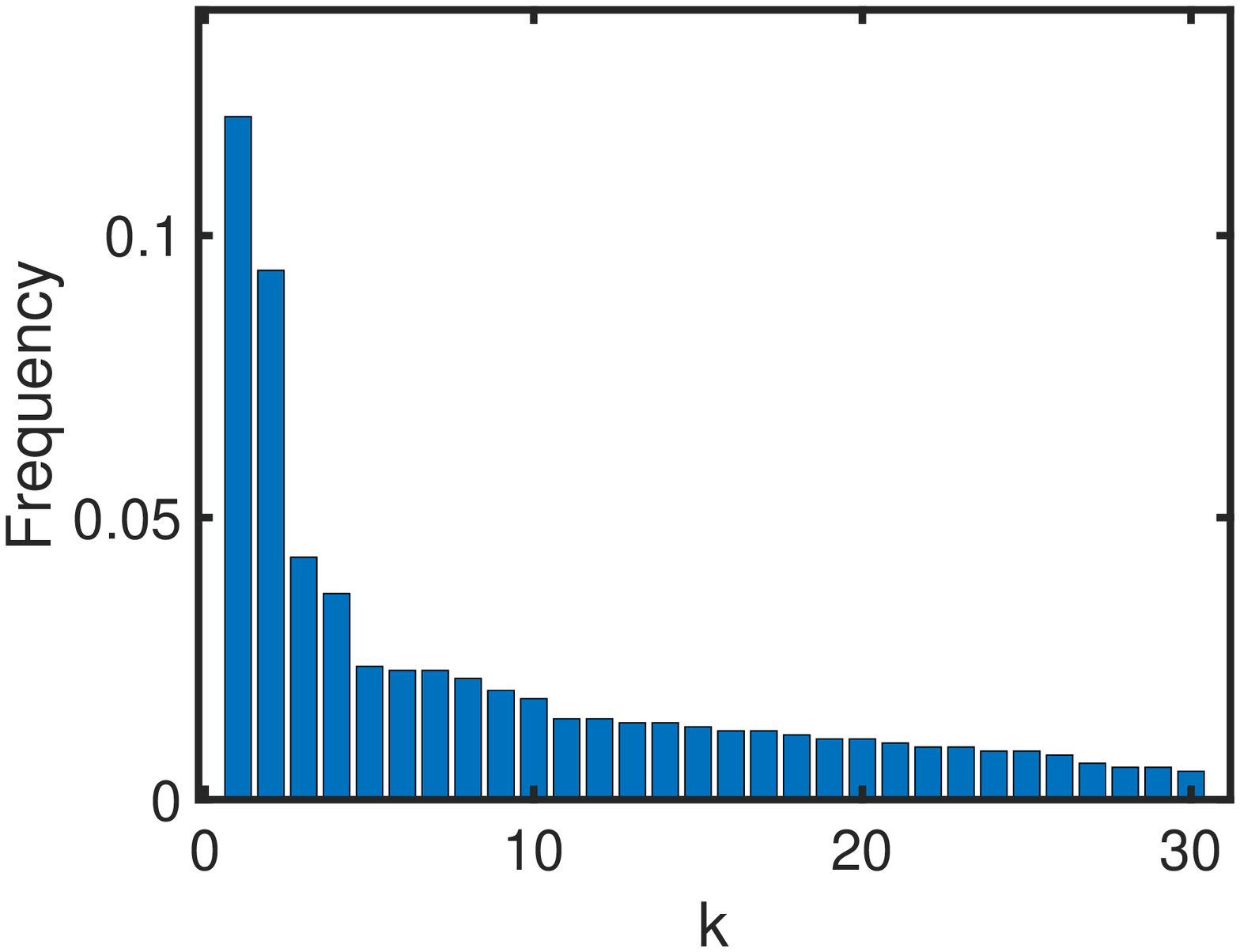}
}
\subfloat[BrightKite]{
\label{Fig-knnmGres}
\includegraphics[scale=0.22]{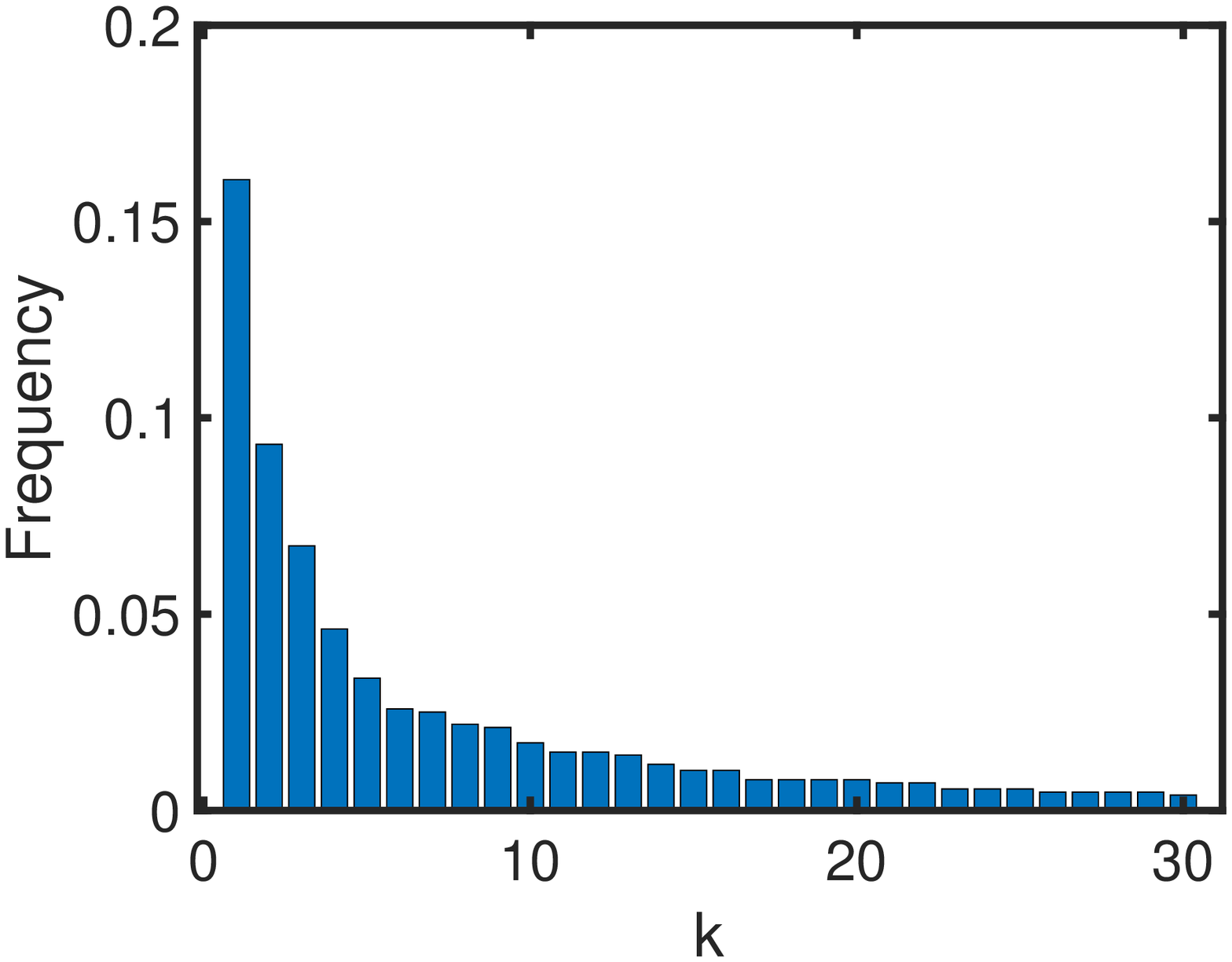}
}

\subfloat[Synthetic]{
\label{Fig-knnmHres}
\includegraphics[scale=0.22]{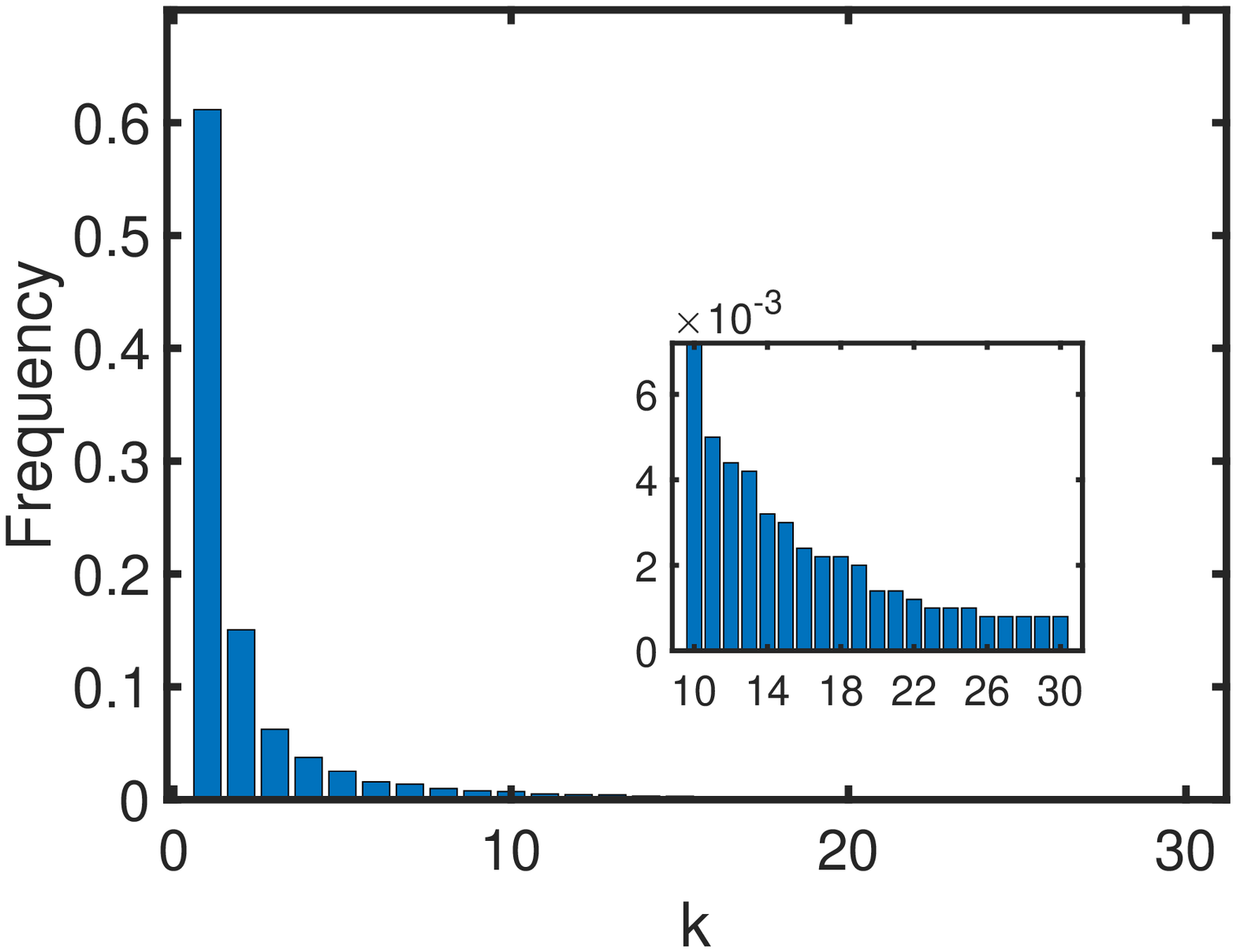}
}
\caption{Frequency of top-$30$ records}
\label{FIG-frequency}
\end{figure}

\textbf{Environment.} All algorithms are implemented in MATLAB, and tested on a machine with Intel Core i5 CPU 2.7Ghz and 8GB RAM. We run each algorithm 10 times, and report the average result.

\textbf{Metric.} We use $\mathrm{Precision}$ to evaluate the effectiveness of the proposed method. 
More specifically, let $P$ be the set of top-$k$ frequent paths statistic based on the original real data,  
and $P'$ be the set of frequent paths found with perturbation. 
The variable Precision measures the fraction of hot paths in the actual query result set are included in the approximated result set, which is shown as follows. 
\begin{equation}
\mathrm{Precision}=\frac{|P\cap P'|}{|P|}, 
\end{equation}

\textbf{Comparison.} We compare our method with Wang et al.'s work \citep{wang2018privtrie}, which is much closer to our work that statistic the frequent items based on trie structure. 
Specifically, they proposed three methods, BSL, IBSL, and PrivTrie. 

\begin{itemize}
\item \textit{BSL.} For BSL method, all the workers are participating in every iteration. 
For each iteration, the workers report the node following the generalized randomized response.
\item \textit{IBSL.} To reduce the privacy budget cost, the workers are partitioned into $|D|$ groups and the workers in each group only participate in one iteration. The reporting process also follows the generalized randomized response. 
\item \textit{PrivTrie.} PrivTrie saves the privacy budget by a careful design of the candidate set construction, which restricts that each user can only report one time to the nodes on his own path. 
To guarantee the accuracy, the size of the candidate set for each node is set to be $b=\max\{\frac{8}{\epsilon^2\alpha^2}, \frac{n}{1000}\}$, where $0<\alpha<1$. When $\epsilon=0.1$, $\min\frac{8}{\epsilon^2\alpha^2}>800$,  which only has small population number of only around $1000$. 
\end{itemize}
Therefore, we compare our methods with both BSL and IBSL in the following experiment.

\textbf{Parameter.} According to paper \citep{wang2018privtrie}, we set the threshold $\theta_i=\frac{\eta \cdot n}{\epsilon_i \sqrt{n^\ast_i}}$ for BSL and IBSL, where $\eta=0.1$ and $n^\ast_i$ is the number of workers participant in each iteration. The threshold for LDPss1 and LDPss2 can be adjusted according to the efficiency of the experiment, as the threshold doesn't affect the performance when it is not very big. 
Regarding privacy parameters, we set $\epsilon_s=\epsilon_r=\frac{1}{2}\epsilon_i$ for LDPss2 in each iteration and $\epsilon_i=\frac{1}{|D|-1}\epsilon$. And we set $\alpha=0.6$. 
Besides, we use $5$ days travel path for both real datasets and synthetic dataset. 
For each setting, average performance is taken over $20$ separate experiments.

\subsection{Impact of the number of queried path k}
We examine the performance of the proposed methods in relation to the number of the queried path $k$ in terms of precision. We vary the number of queried path between $1$ and $20$ in Step 1 on three datasets. 
The result includes $4$ different values for privacy budget $\epsilon$. 

\begin{figure*}[tbp]
\centering

\subfloat[Gowalla:~$\epsilon=0.1$ ]{
\label{Fig-Gowalla01}
\includegraphics[scale=0.22]{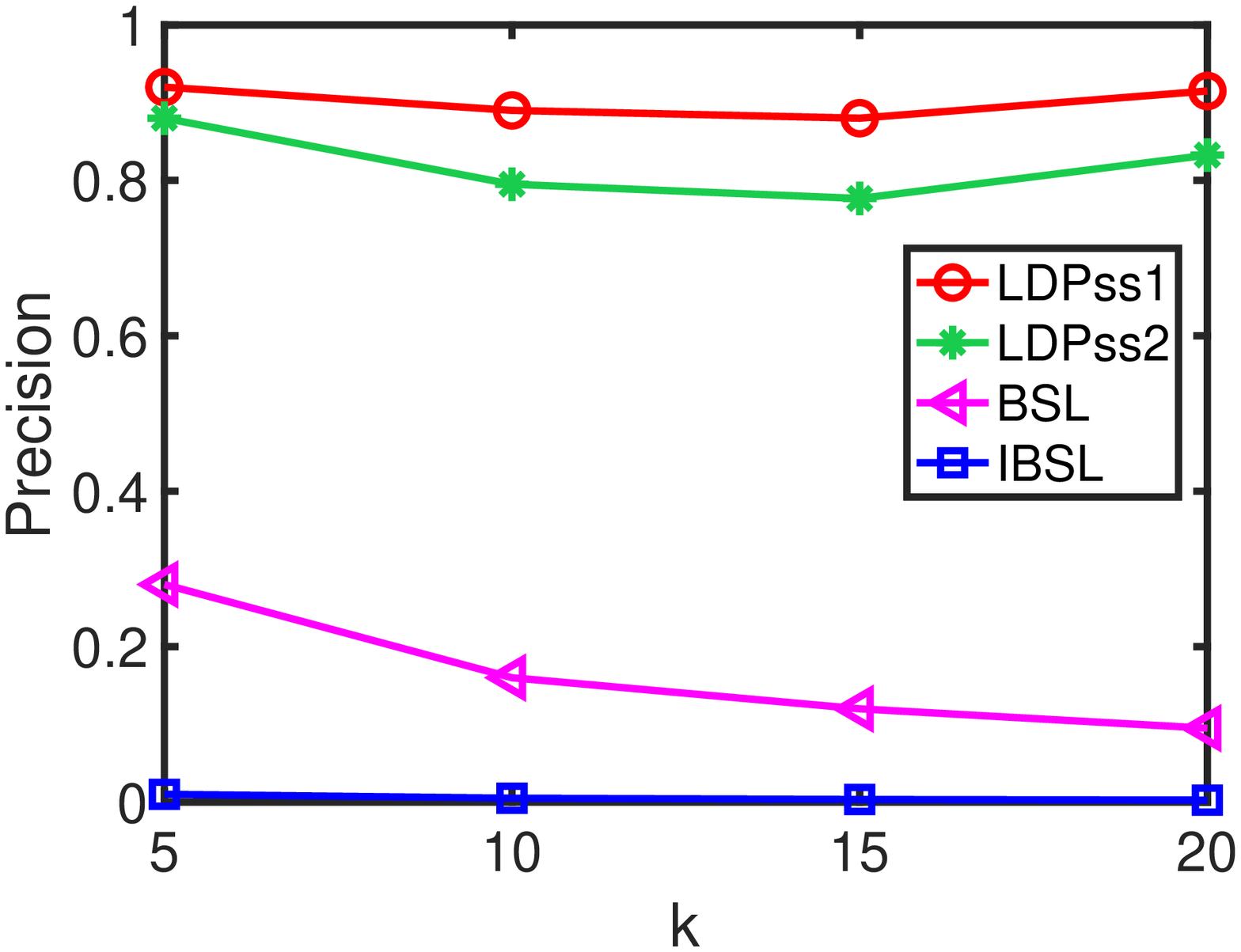}
}
\subfloat[Gowalla:~$\epsilon=0.5$ ]{
\label{Fig-knnmGres1}
\includegraphics[scale=0.22]{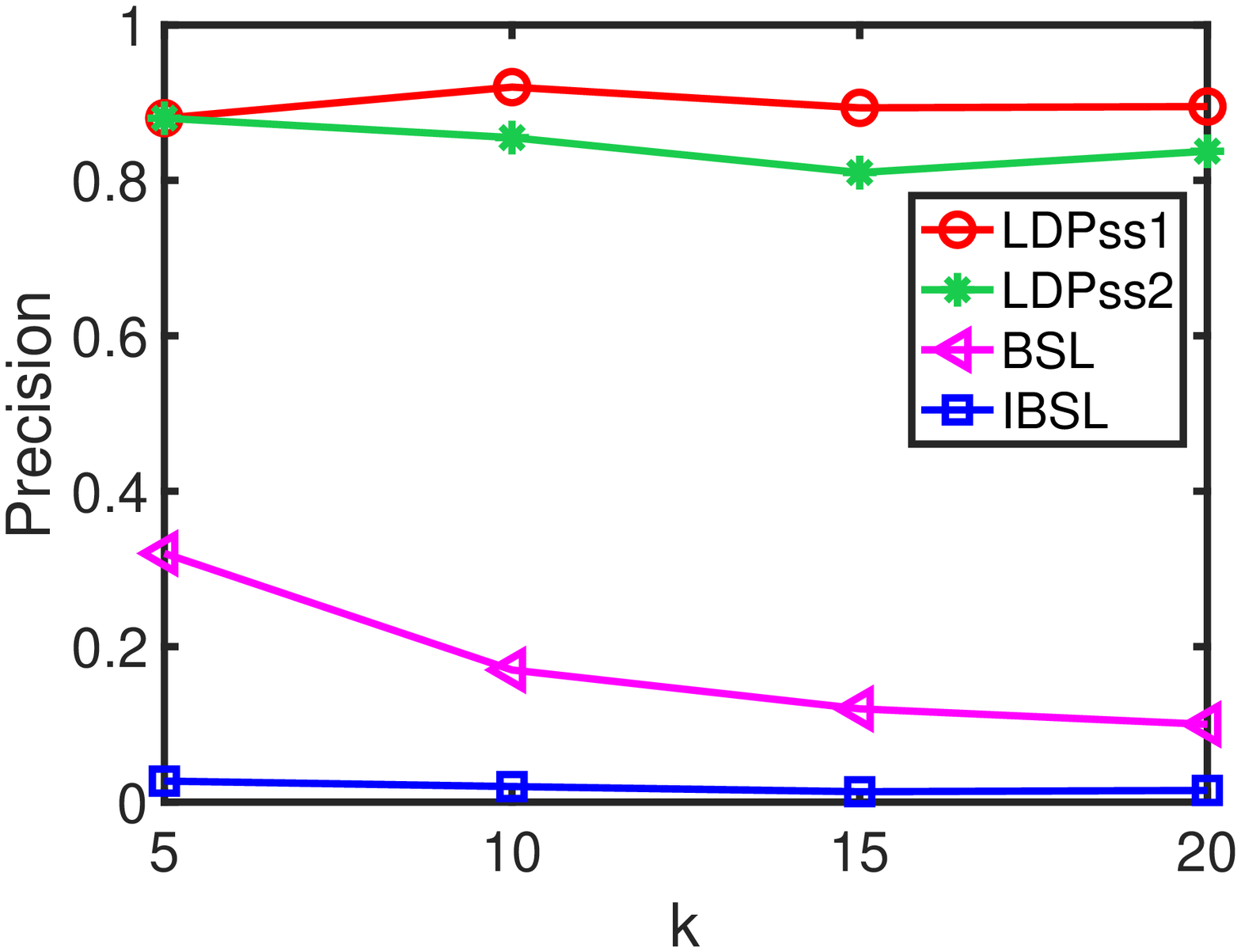}
}
\subfloat[Gowalla:~$\epsilon=1$ ]{
\label{Fig-knnmHres1}
\includegraphics[scale=0.22]{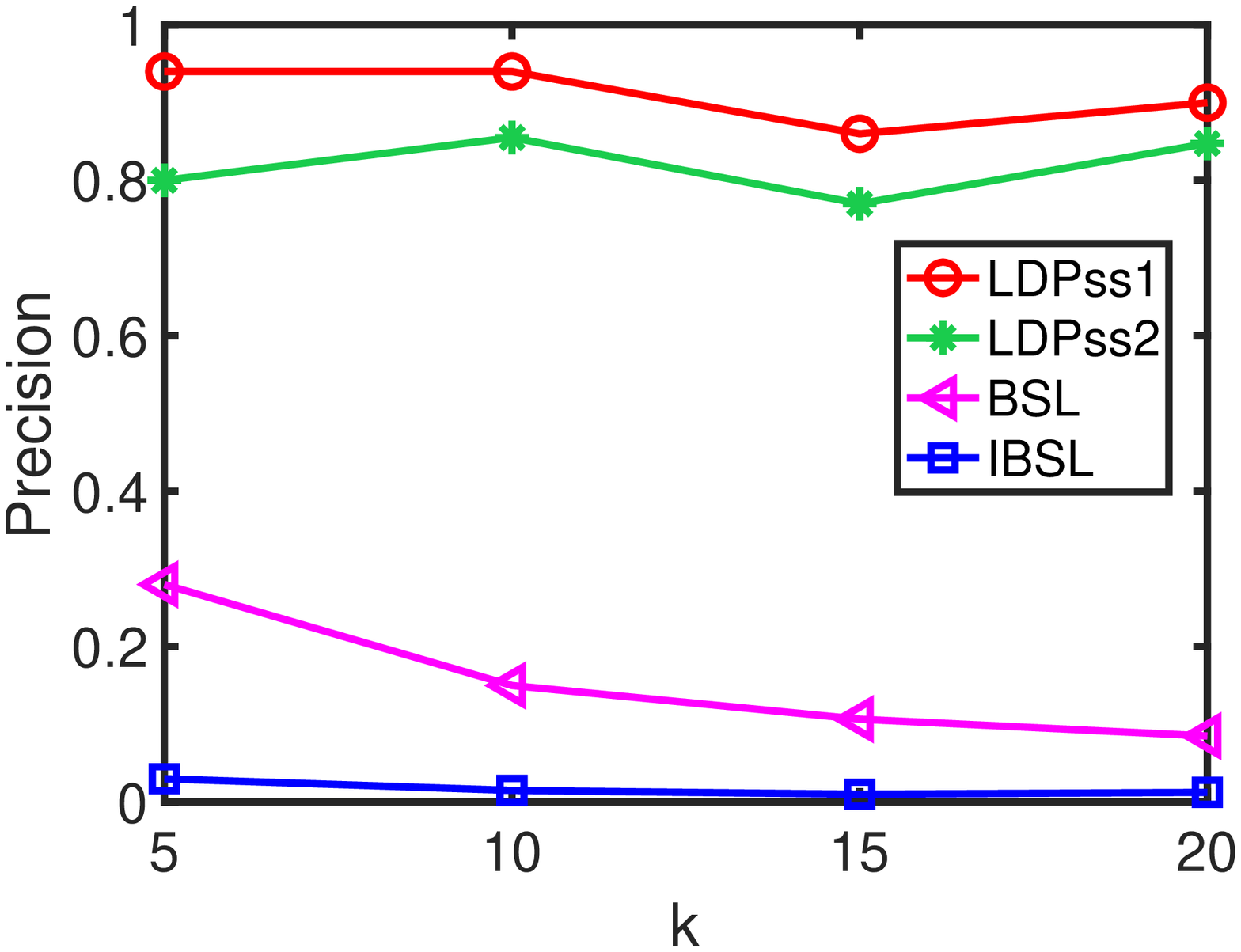}
}
\subfloat[Gowalla:~$\epsilon=2$ ]{
\label{Fig-knnmCres}
\includegraphics[scale=0.22]{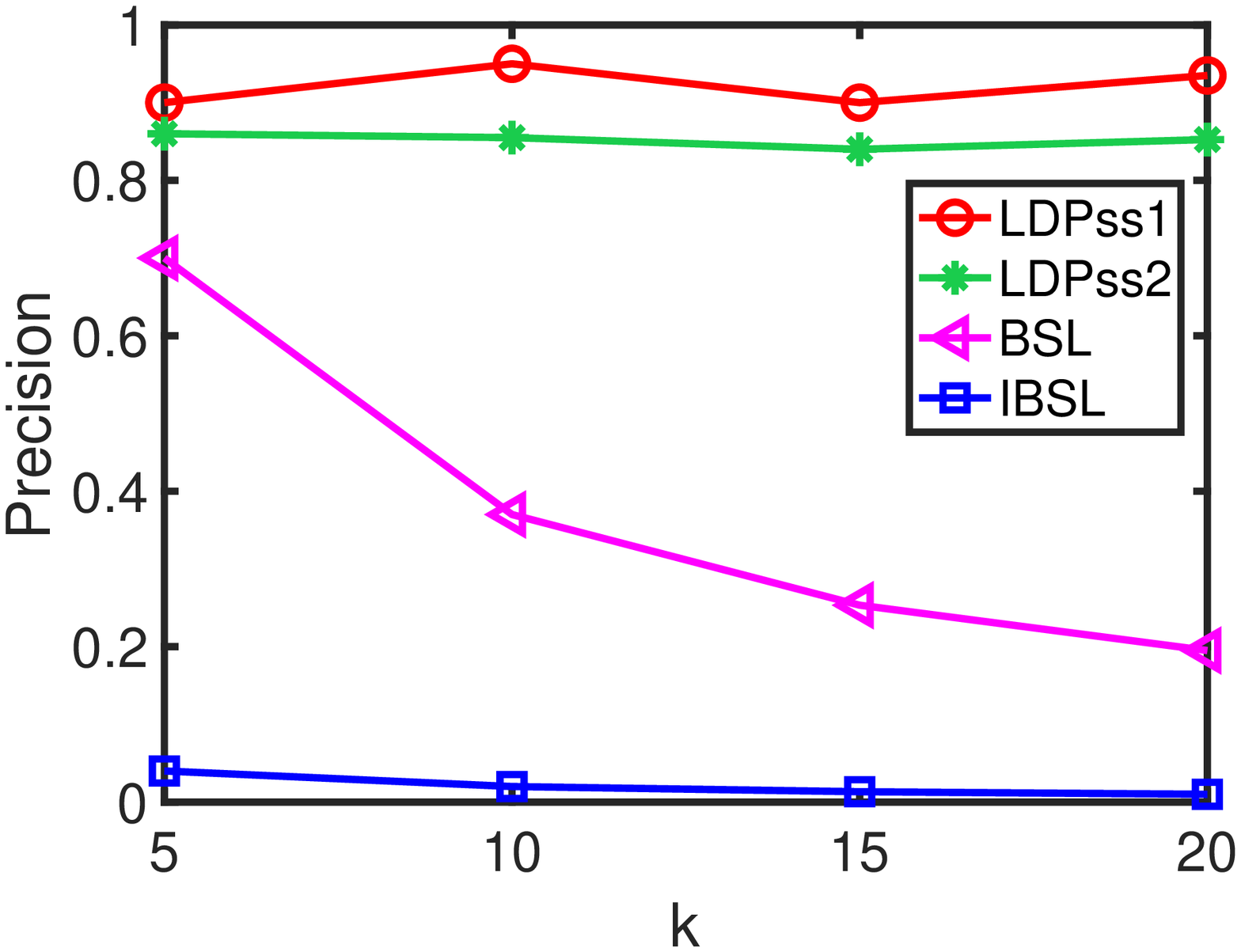}
}

\subfloat[BrightKite:~$\epsilon=0.1$ ]{
\label{Fig-B01}
\includegraphics[scale=0.22]{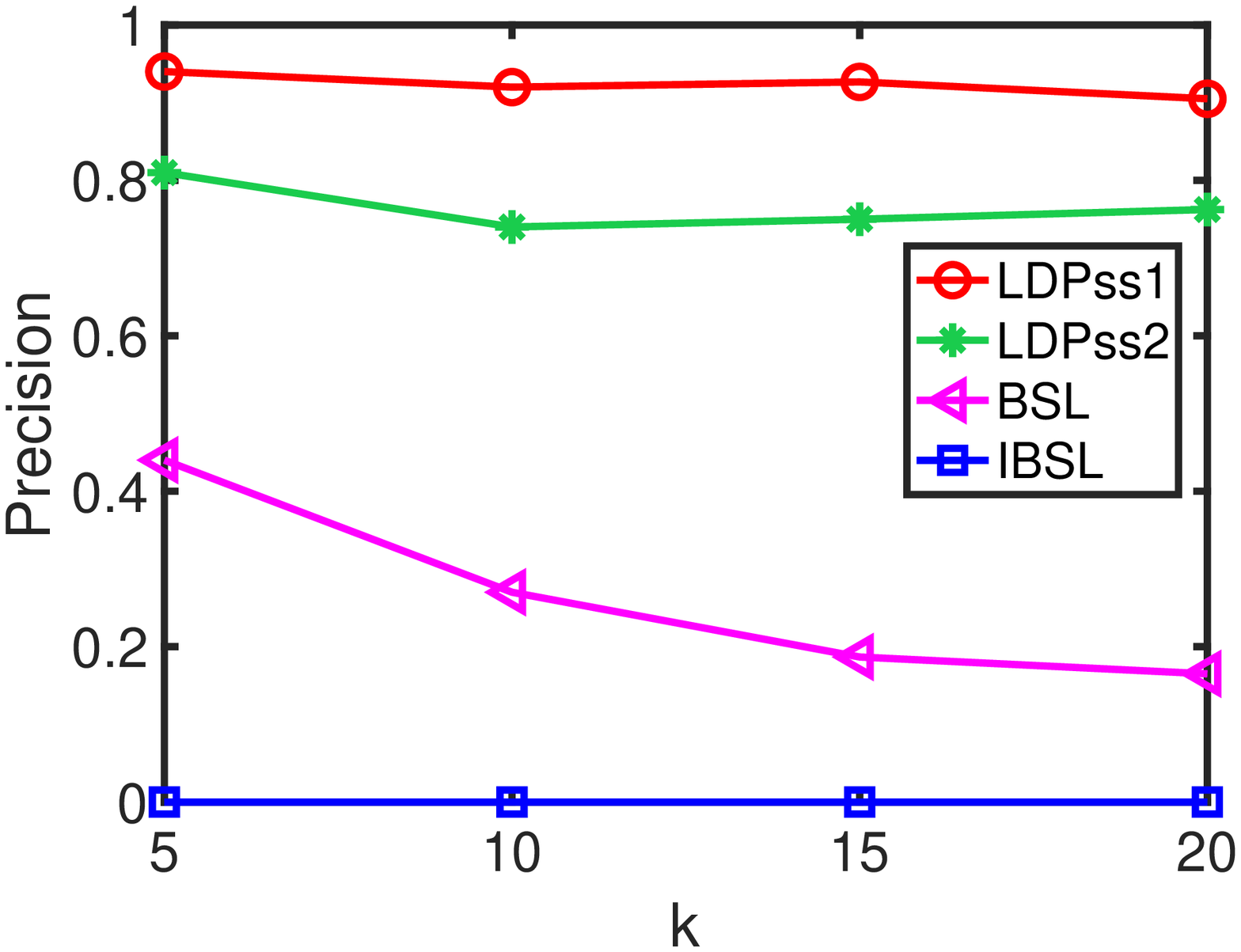}
}
\subfloat[BrightKite:~$\epsilon=0.5$]{
\label{Fig-knnmGres2}
\includegraphics[scale=0.22]{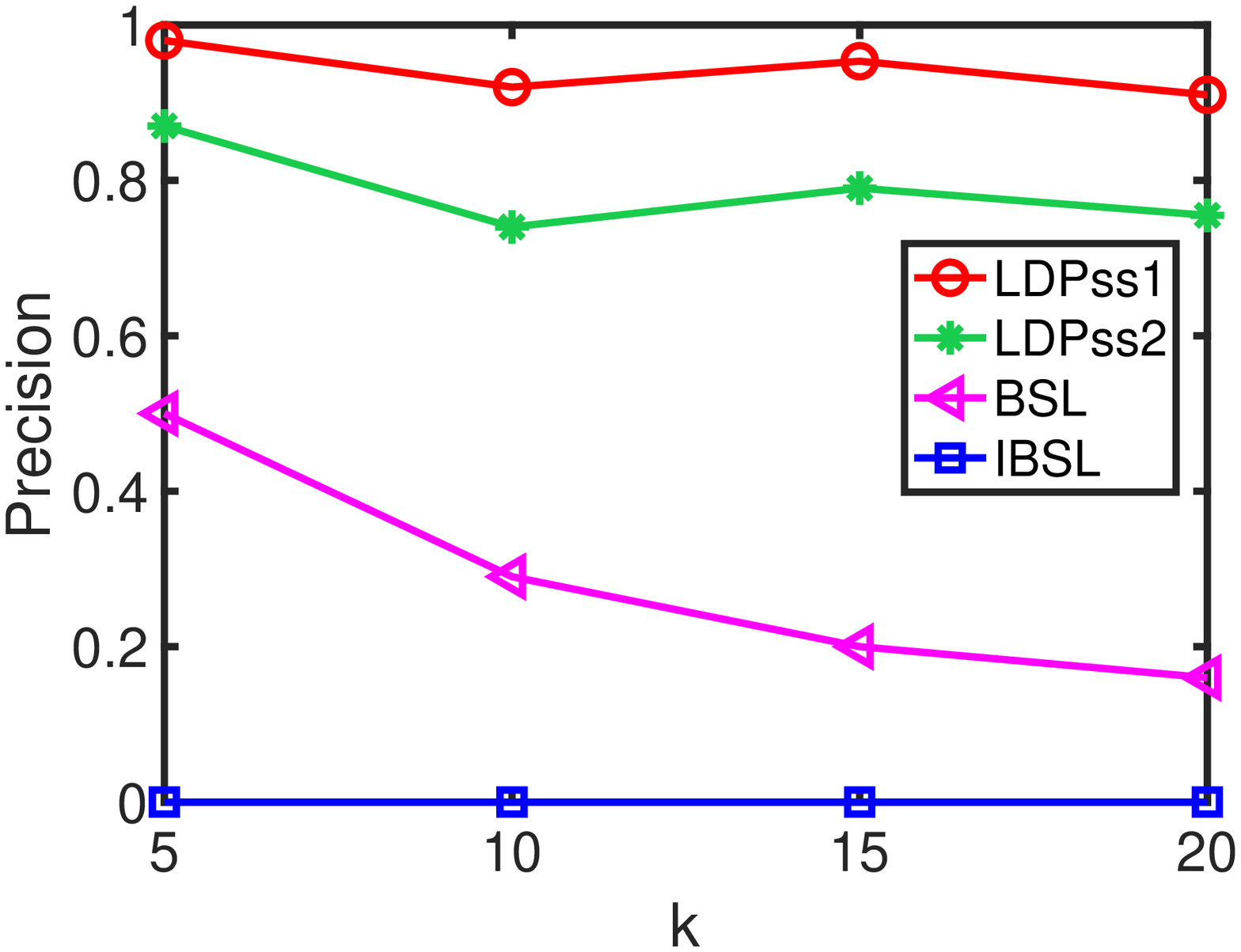}
}
\subfloat[BrightKite:~$\epsilon=1$]{
\label{Fig-knnmHres2}
\includegraphics[scale=0.22]{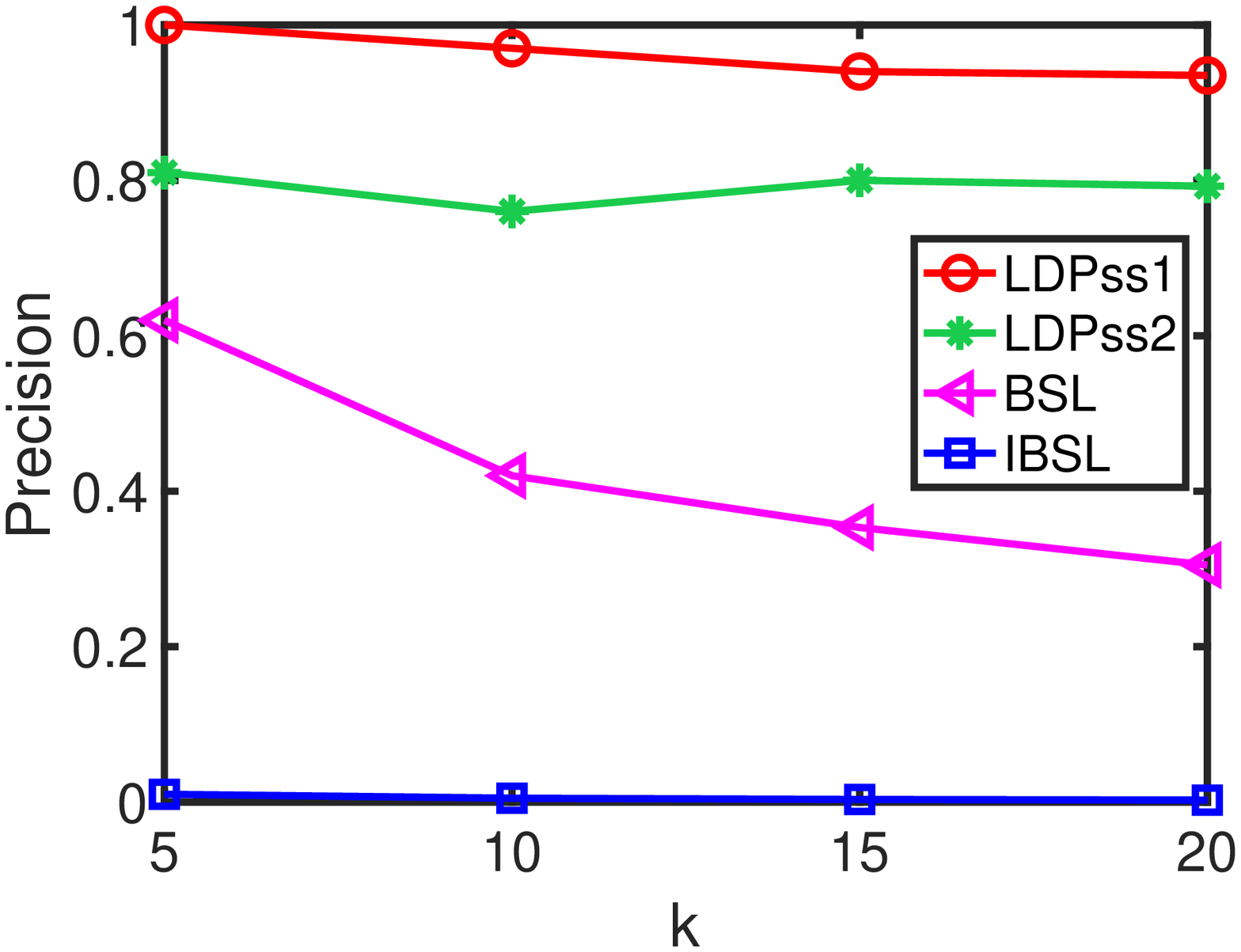}
}
\subfloat[BrightKite:~$\epsilon=2$]{
\label{Fig-knnmCres1}
\includegraphics[scale=0.22]{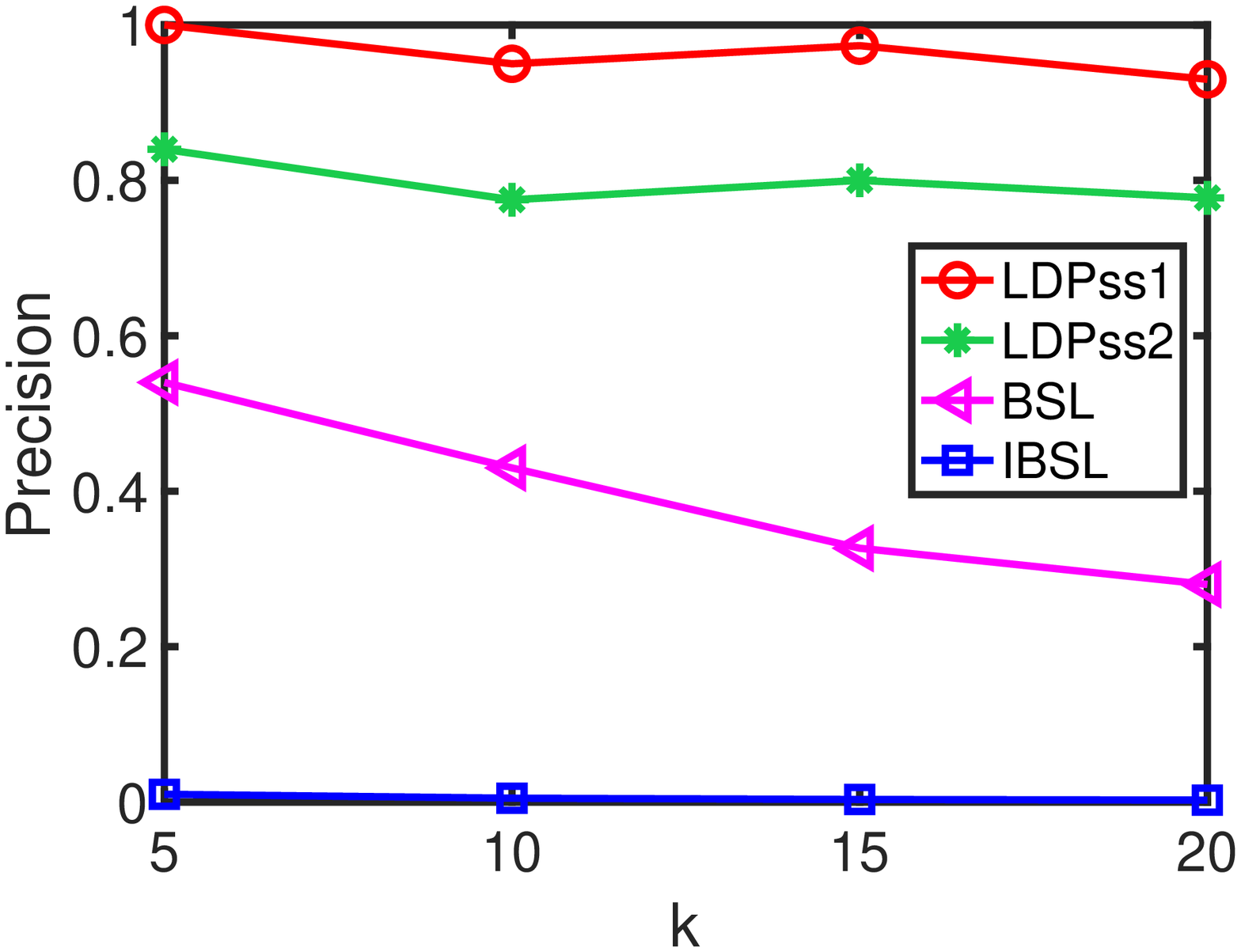}
}

\subfloat[Synthetic:~$\epsilon=0.1$]{
\label{Fig-zip01}
\includegraphics[scale=0.22]{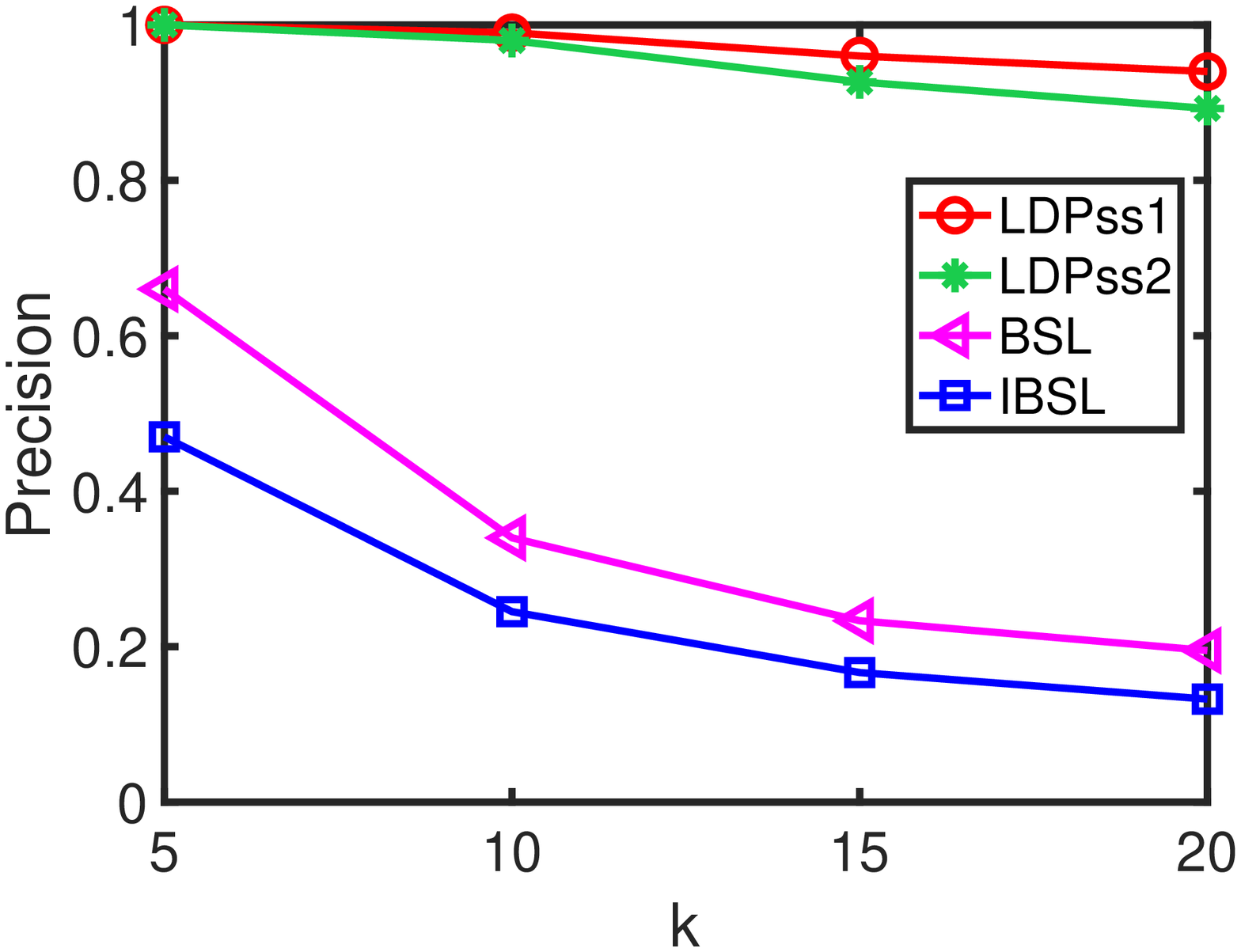}
}
\subfloat[Synthetic:~$\epsilon=0.5$]{
\label{Fig-knnmGres3}
\includegraphics[scale=0.22]{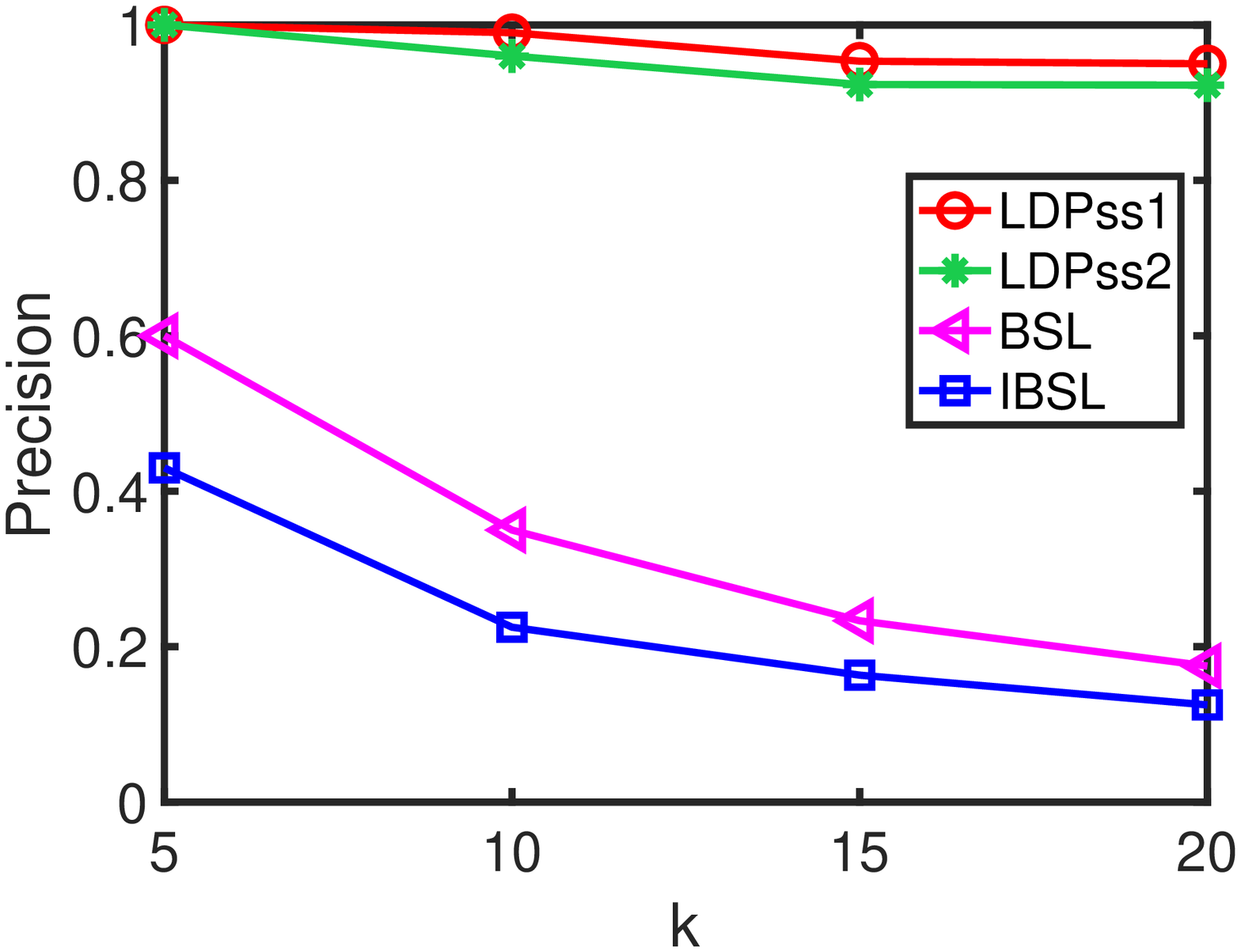}
}
\subfloat[Synthetic:~$\epsilon=1$]{
\label{Fig-knnmHres3}
\includegraphics[scale=0.22]{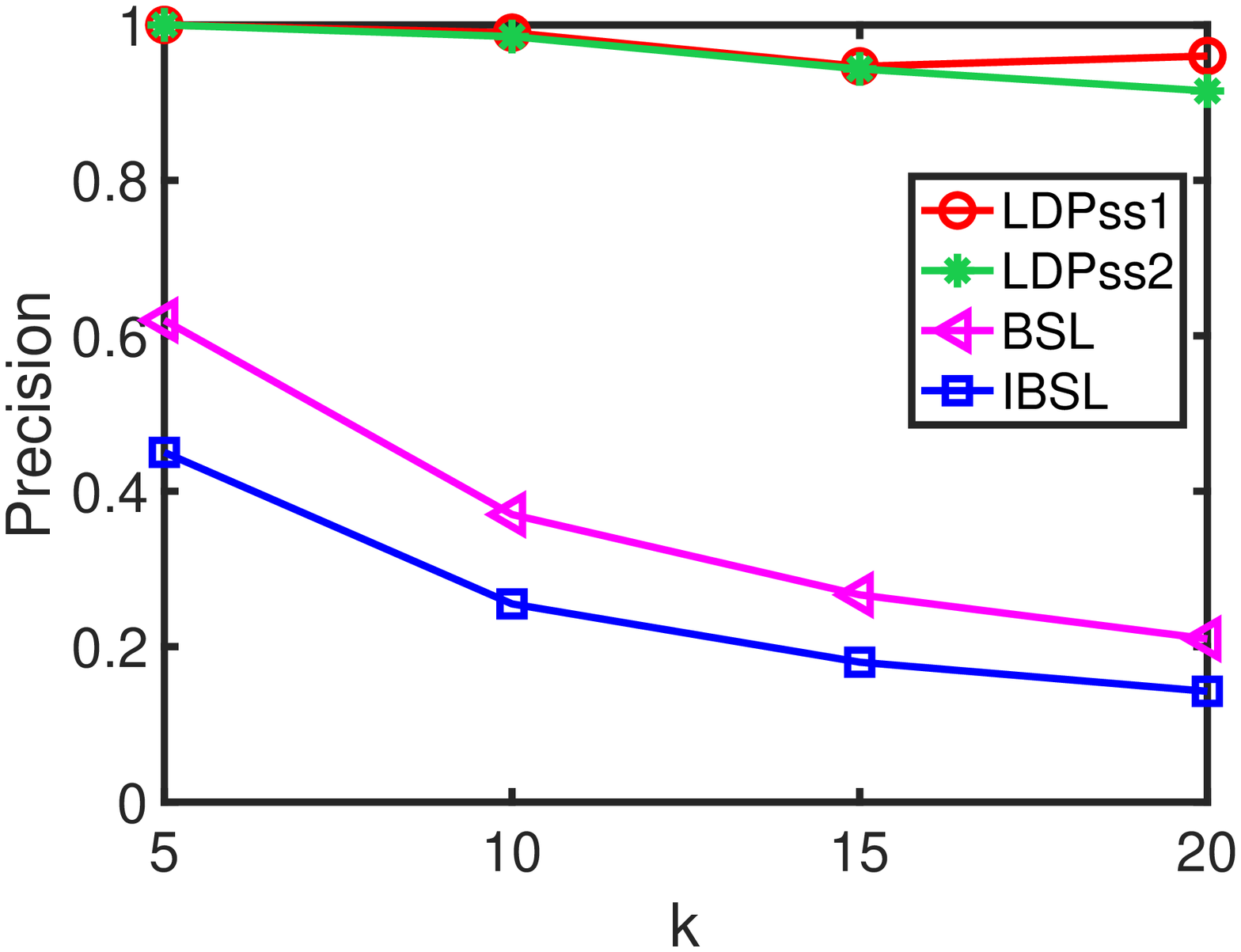}
}
\subfloat[Synthetic:~$\epsilon=2$]{
\label{Fig-knnmCres2}
\includegraphics[scale=0.22]{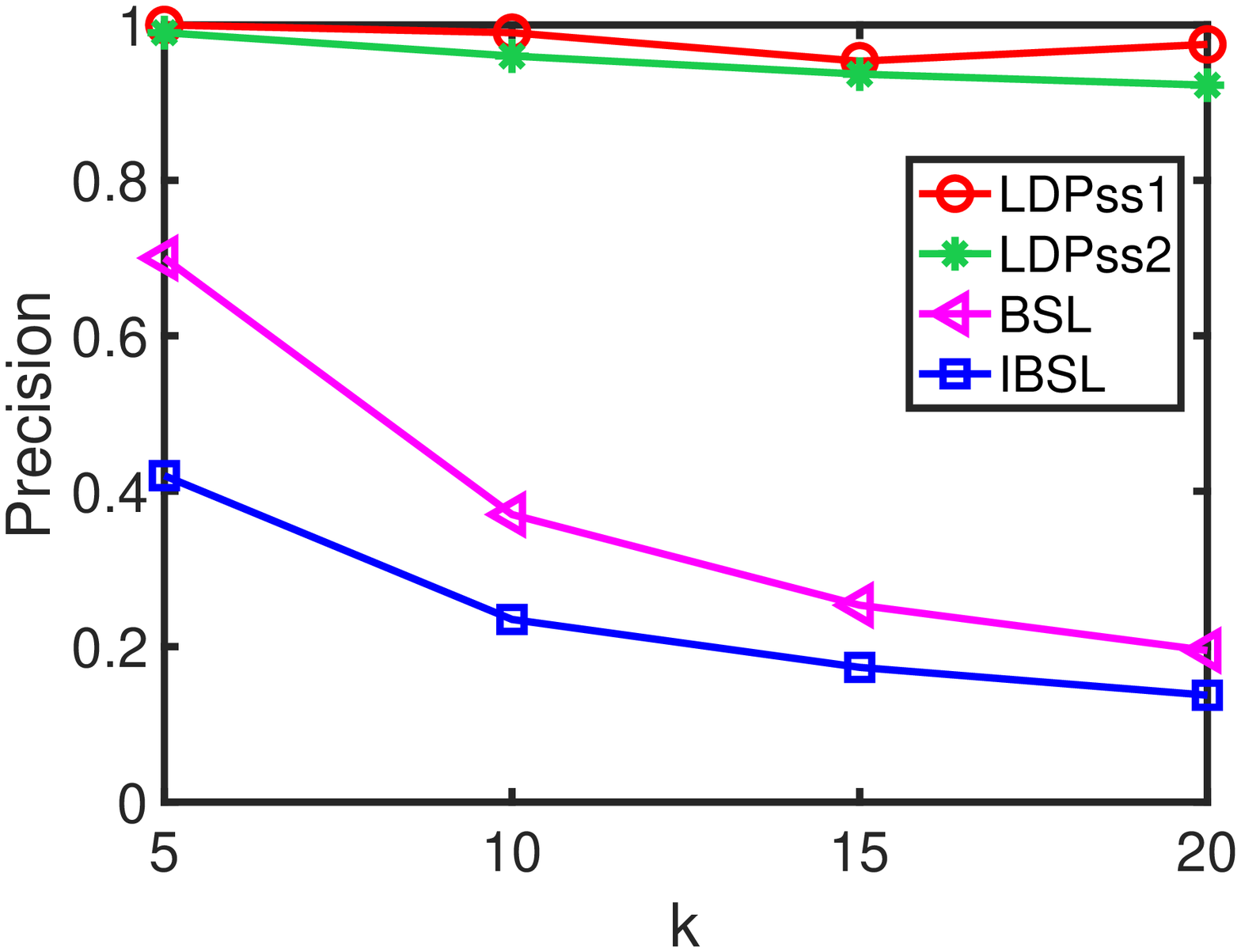}
}

\caption{Effect of $k$ on three datasets}
\label{FIG-k}
\end{figure*}

Fig. \ref{FIG-k} shows the precision values corresponding to different values of $k$ on Gowalla, BrightKite and synthetic datasets. 
It is clear that both LDPss1 and LDPss2 outperform BSL and IBSL in all configurations. 
For example, in Fig. \ref{Fig-Gowalla01} on Gowalla dataset, the precision of LDPss1 is around $0.9$, which outperforms BSL by around $60\%$ when $k=5$, and outperforms IBSL by almost $90\%$. LDPss2 has a similar result. 
Also, we find that the IBSL has a very poor performance, especially on Gowalla and BrightKite datasets. 
This is because only around $200$ workers participate in the protocol in each iteration. 
Though each user is assigned a more privacy budget compared to BSL, 
it is not enough to make much difference when the data dimension is high. 
The performance of IBSL is much better on Synthetic dataset because more workers participate which reduces the statistical error. 

In addition, the number of queried frequent path has little effect on the precision of the proposed methods while the precision of the other two methods reduce significantly with the increase of $k$. 
This is caused by the fact that in BSL and IBSL, as $k$ increases, the frequency distribution between locations become closer to each other and hence harder to distinguish between one and another.
This phenomenon can be verified on the synthetic dataset, which has an obvious frequency difference when $k<5$. 
The precision decreases dramatically to around $20\%$ for both BSL and IBSL, but the precision of the proposed method still achieves precision over $90\%$. A similar result is also observed when $\epsilon$ has different values in Fig. \ref{FIG-k}.

\subsection{Performance with different privacy budgets}
We examined the performance of the three methods in relation to the different privacy
budgets $\epsilon$ on three datasets for different values of $\epsilon\in\{0.1,0.5,1,2\}$. 
We query top-10 frequent paths. 

\begin{figure*}[tbp]
\centering

\subfloat[Gowalla]{
\label{Fig-epsilonG}
\includegraphics[scale=0.22]{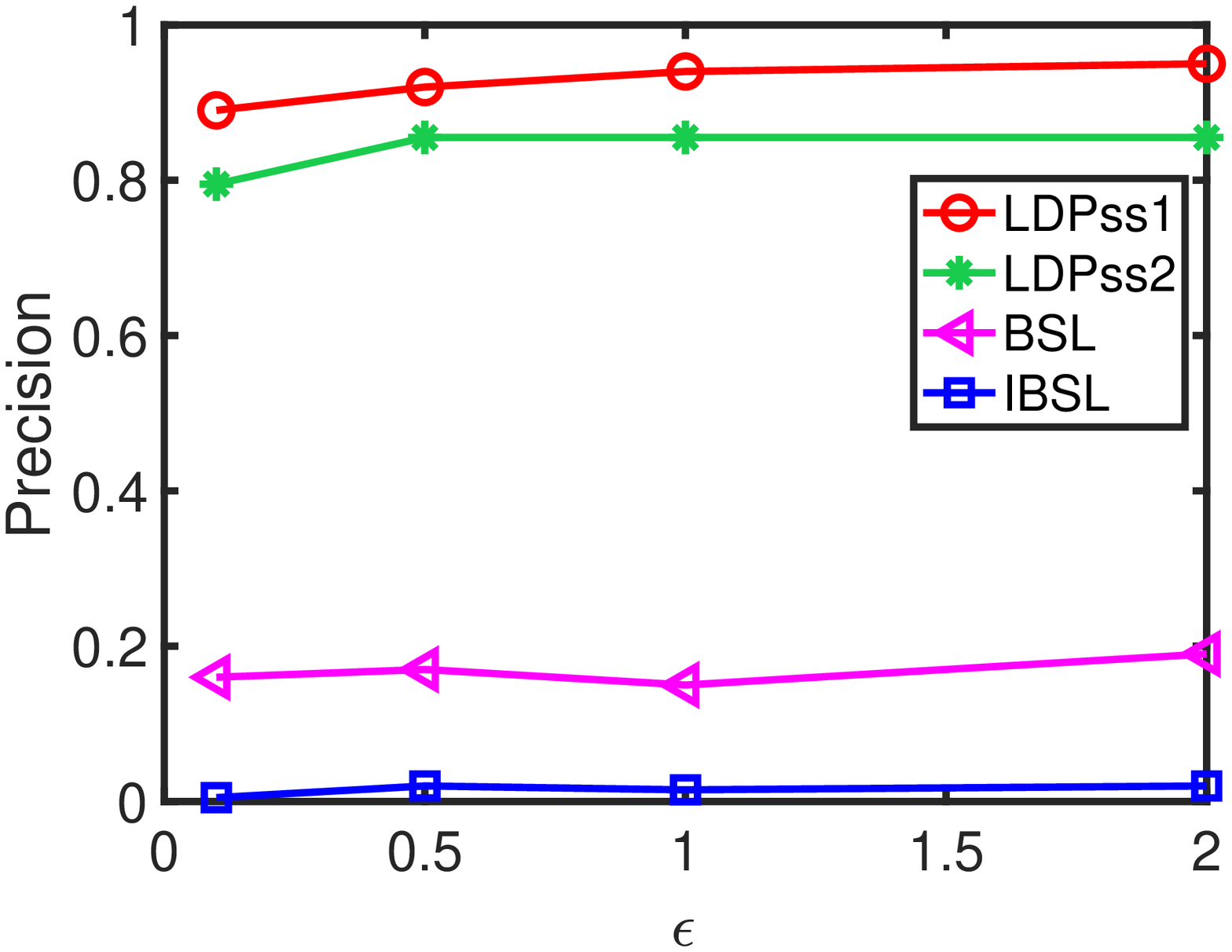}
}
\subfloat[BrigthKite]{
\label{Fig-knnmGres4}
\includegraphics[scale=0.22]{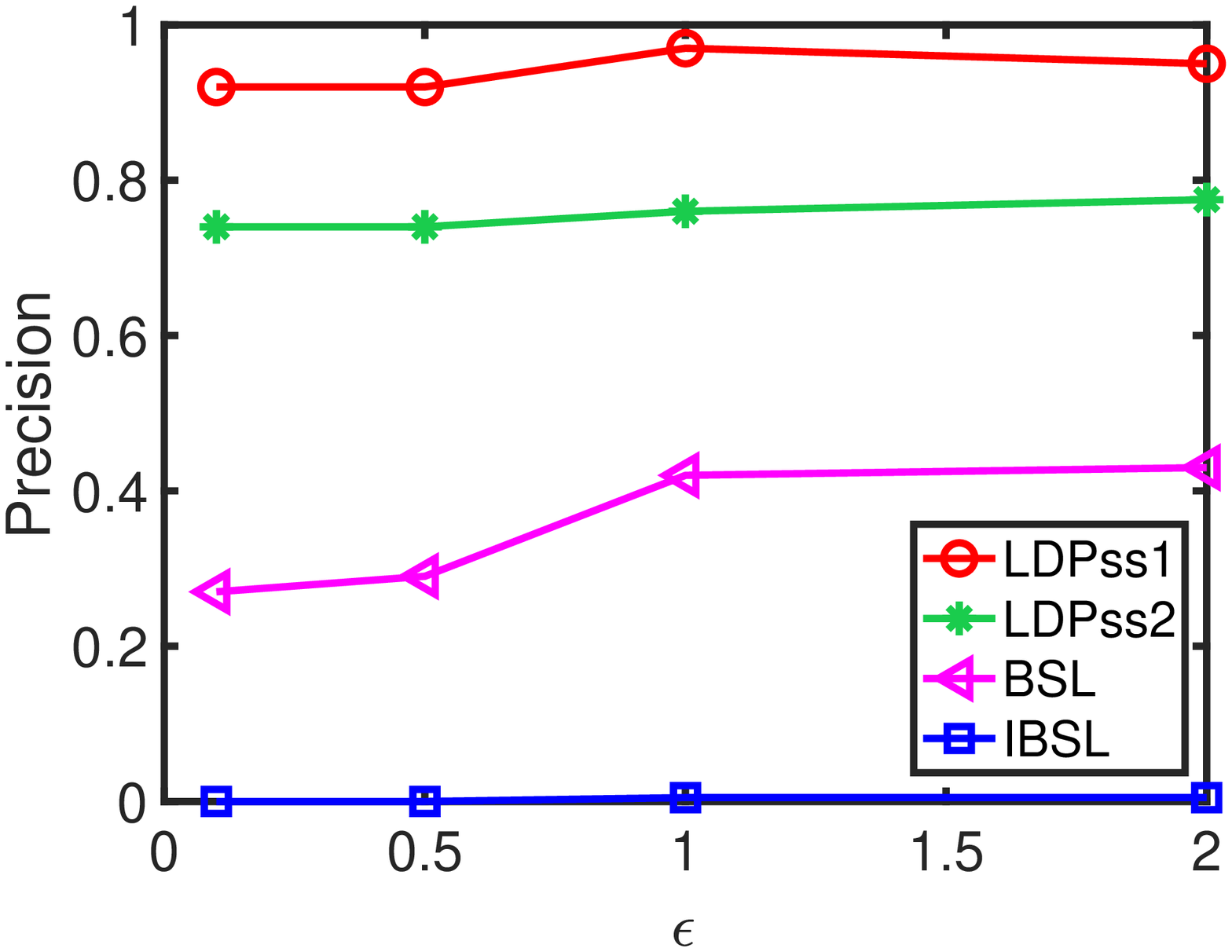}
}
\subfloat[Synthetic]{
\label{Fig-knnmHres4}
\includegraphics[scale=0.22]{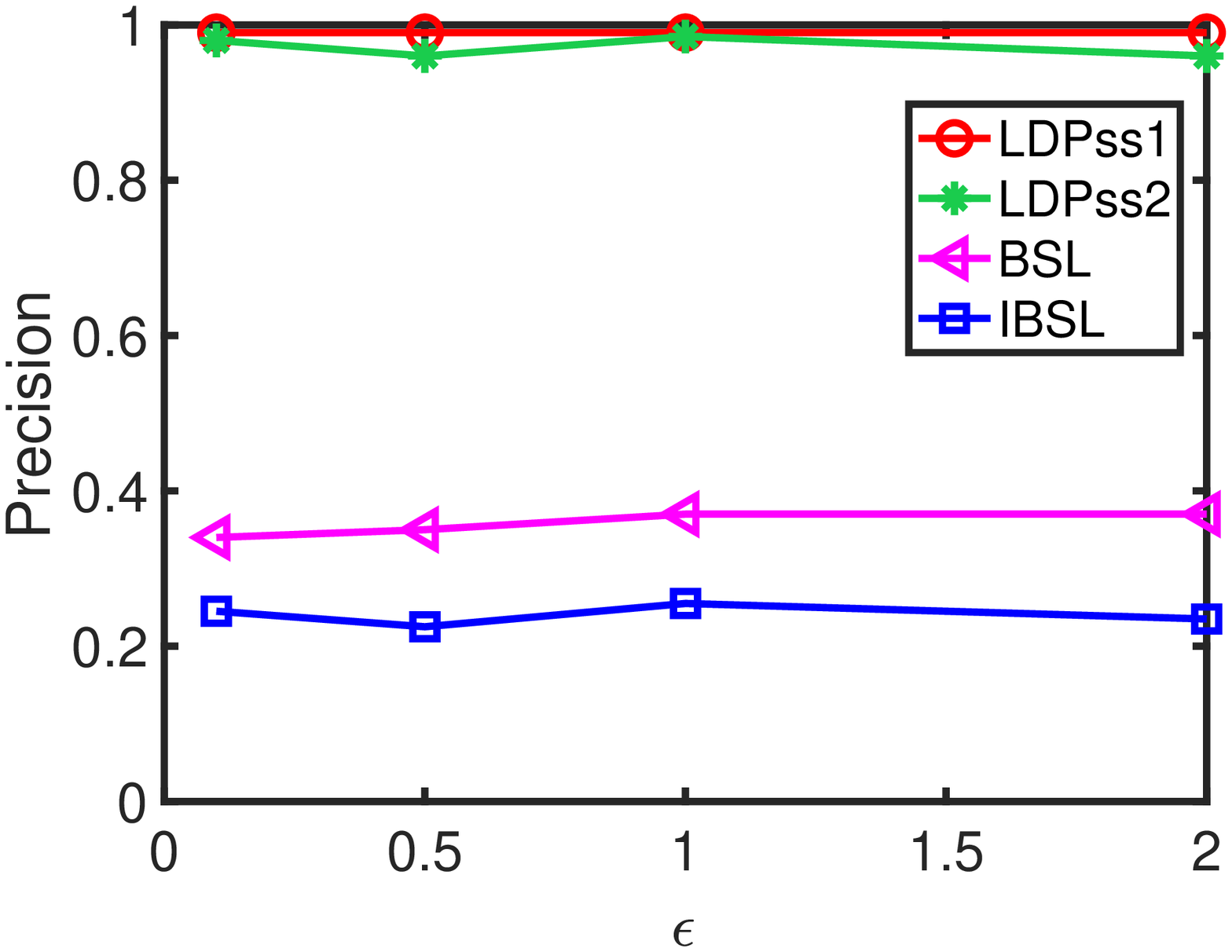}
}
\caption{Effect of $\epsilon$ on three datasets}
\label{FIG-epsilon}
\end{figure*}

Fig. \ref{FIG-epsilon} shows the change in precision with a varied privacy budget.
We observed that the proposed methods always outperform the other two methods and the precision increased as the privacy budget $\epsilon$ increased for the proposed methods on all the three datasets. This is because a bigger privacy budget $\epsilon$ means a higher proportion of the workers are chosen to report their true value for the proposed methods. 
Specifically, as shown in Fig. \ref{Fig-epsilonG}, when $\epsilon=0.1$, LDPss1 achieves precision at $0.88$ and LDPss2 achieves $0.79$ respectively. 
When $\epsilon=2$, the precision of LDPss1 and LDPss2 achieves $0.93$ and $0.85$ respectively. 
In addition, we find that the increase of the privacy budget does not have a big improvement for both BSL and IBSL. There are two reasons. First, for the traditional local differential privacy, a large number of participants are needed to reduce the error. However, in the experiment, we only have around one thousand records. Second, the dimension of the data is very high in the proposed application. 
For $i$th iteration, the node can only get $\frac{\epsilon}{d_i|D|}$ privacy budget. 
Due to this very small multiplier, small increase of $\epsilon$ does not lead to large difference in the performance. 
A similar result can be found on other two datasets.  

\subsection{Performance with different number of participants}

To examine the effect of the number of participants, 
we randomly sample $1000$, $3000$ and $5000$ records from the synthetic dataset and test the query precision on these three datasets with different privacy budgets. We still set $k=10$.

\begin{figure*}[!tp]
\centering

\subfloat[$\epsilon=0.1$]{
\label{Fig-numberP01}
\includegraphics[scale=0.22]{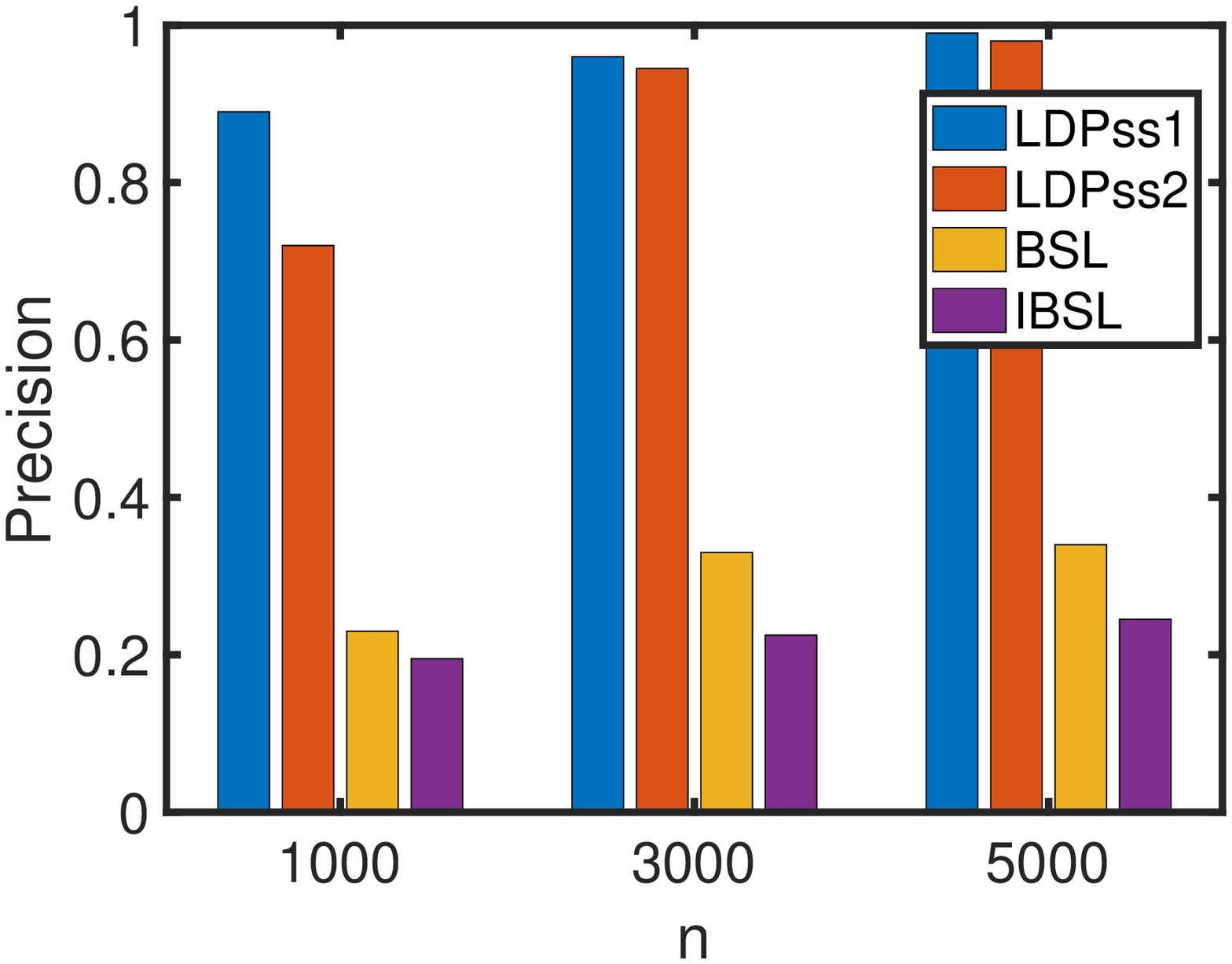}
}
\subfloat[$\epsilon=0.5$]{
\label{Fig-knnmGres5}
\includegraphics[scale=0.22]{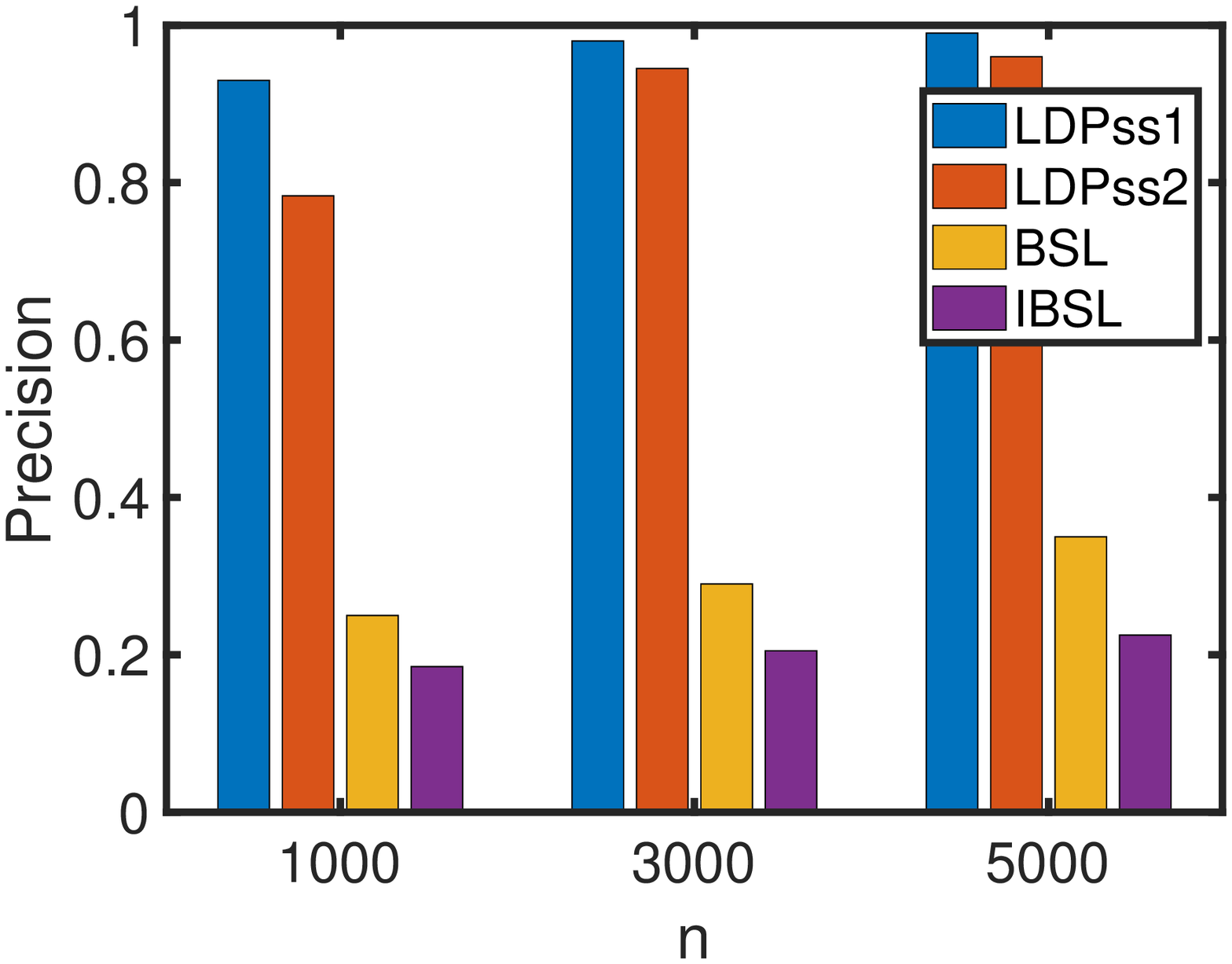}
}
\subfloat[$\epsilon=1$]{
\label{Fig-knnmHres5}
\includegraphics[scale=0.22]{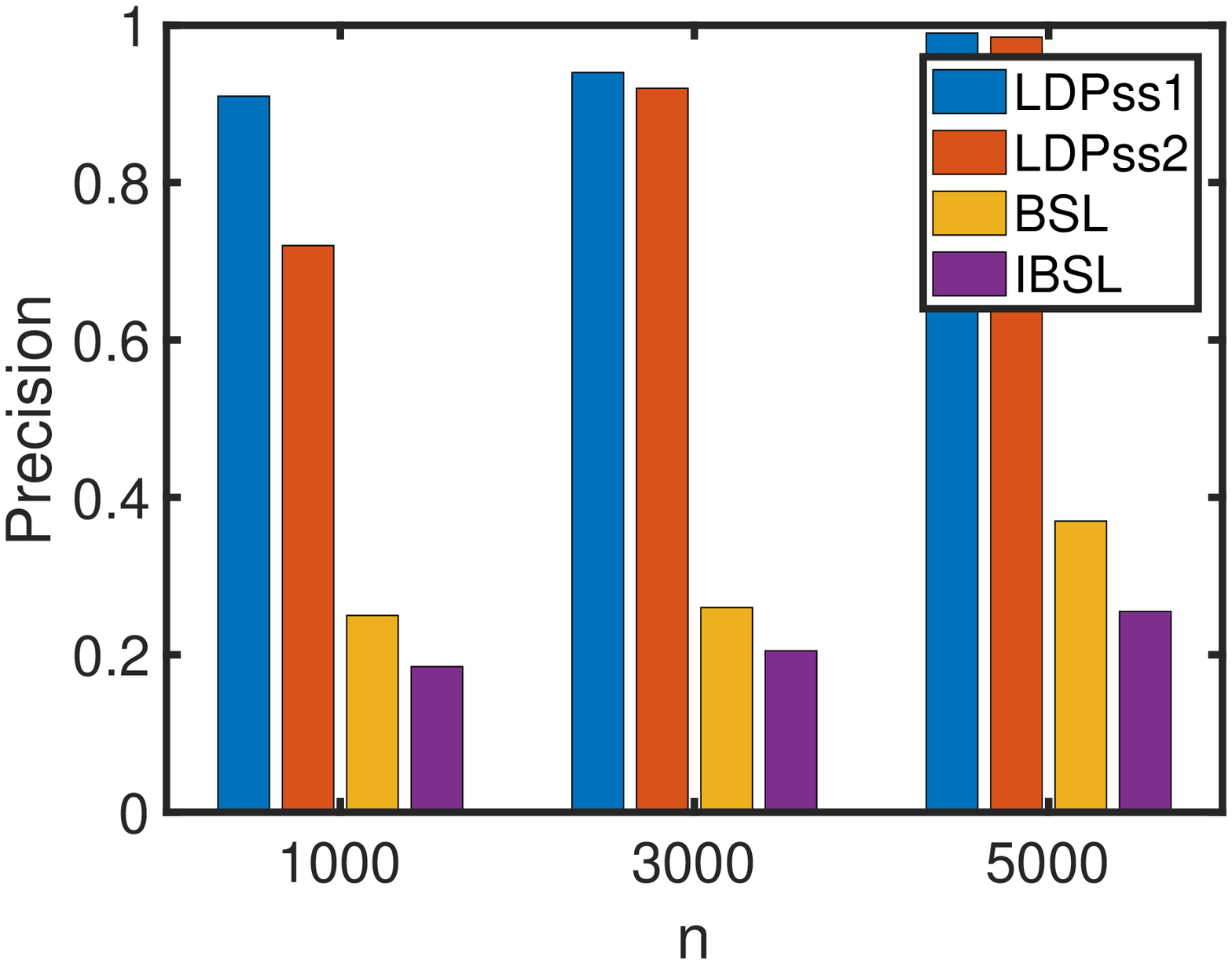}
}
\subfloat[$\epsilon=2$]{
\label{Fig-numberP20}
\includegraphics[scale=0.22]{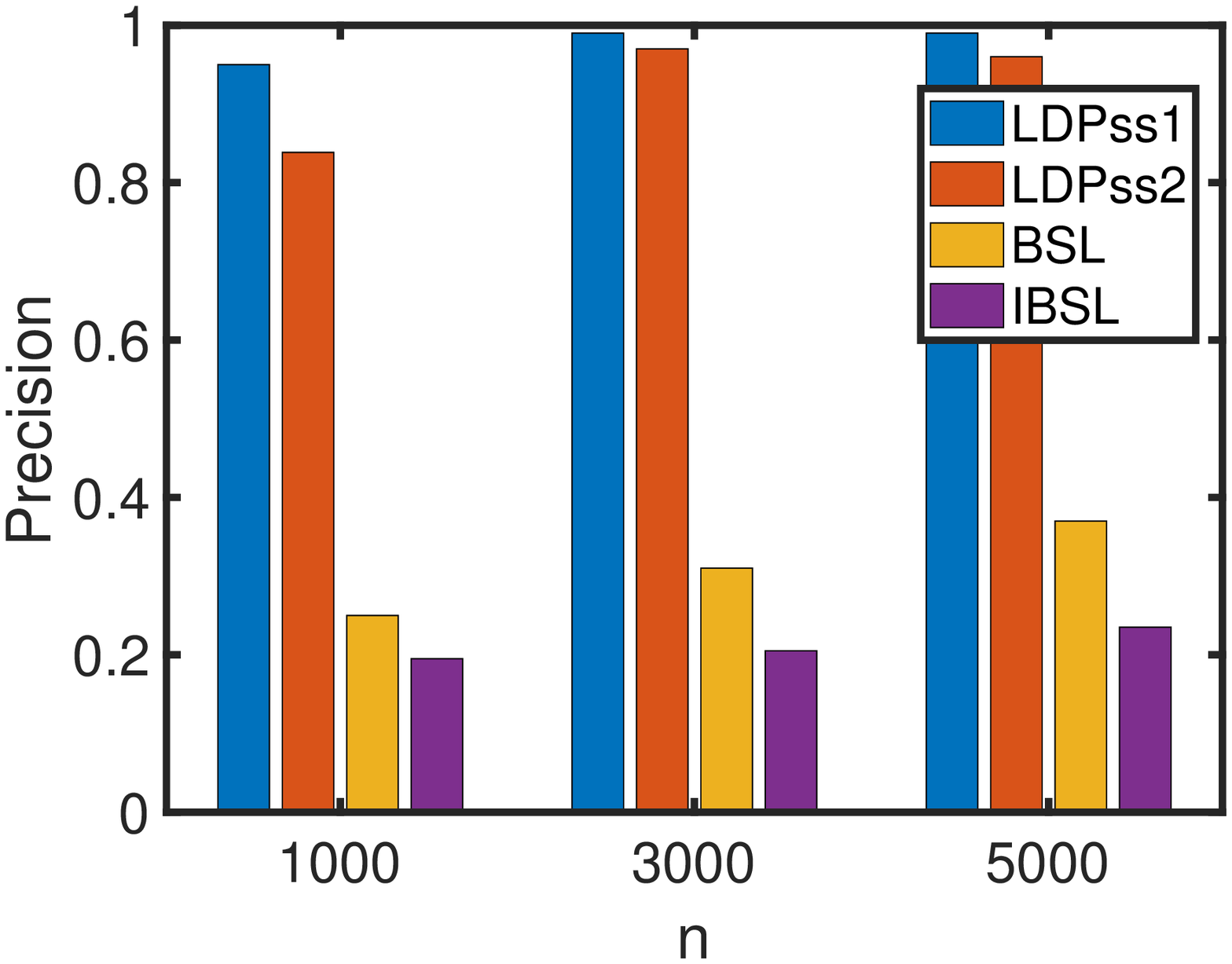}
}
\caption{Effect of the number of participants on the synthetic dataset}
\label{FIG-numberP}
\end{figure*}

\begin{figure*}[tbp]
\centering

\subfloat[$\epsilon=0.1$]{
\label{Fig-locationG}
\includegraphics[scale=0.22]{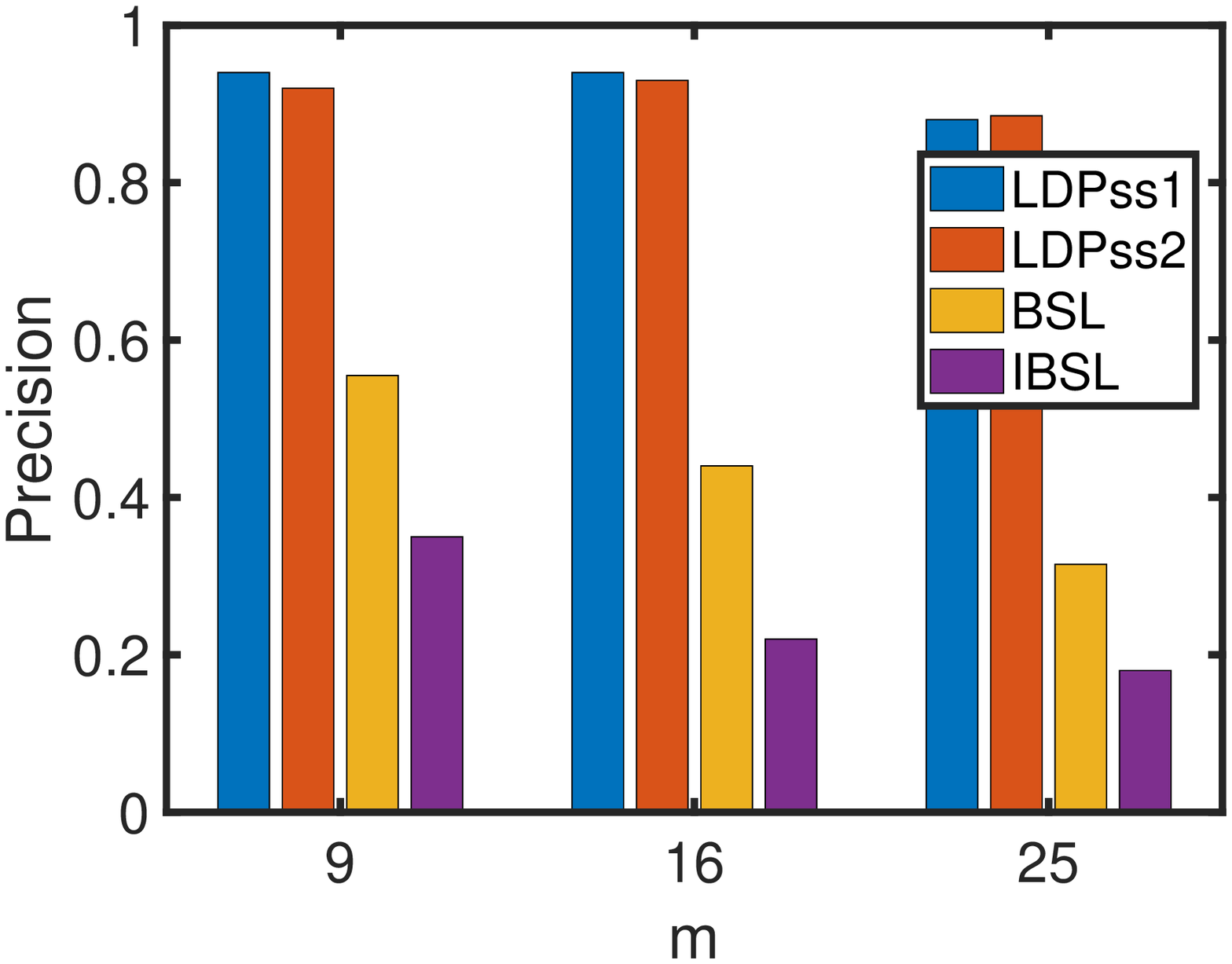}
}
\subfloat[$\epsilon=0.5$]{
\label{Fig-locationB}
\includegraphics[scale=0.22]{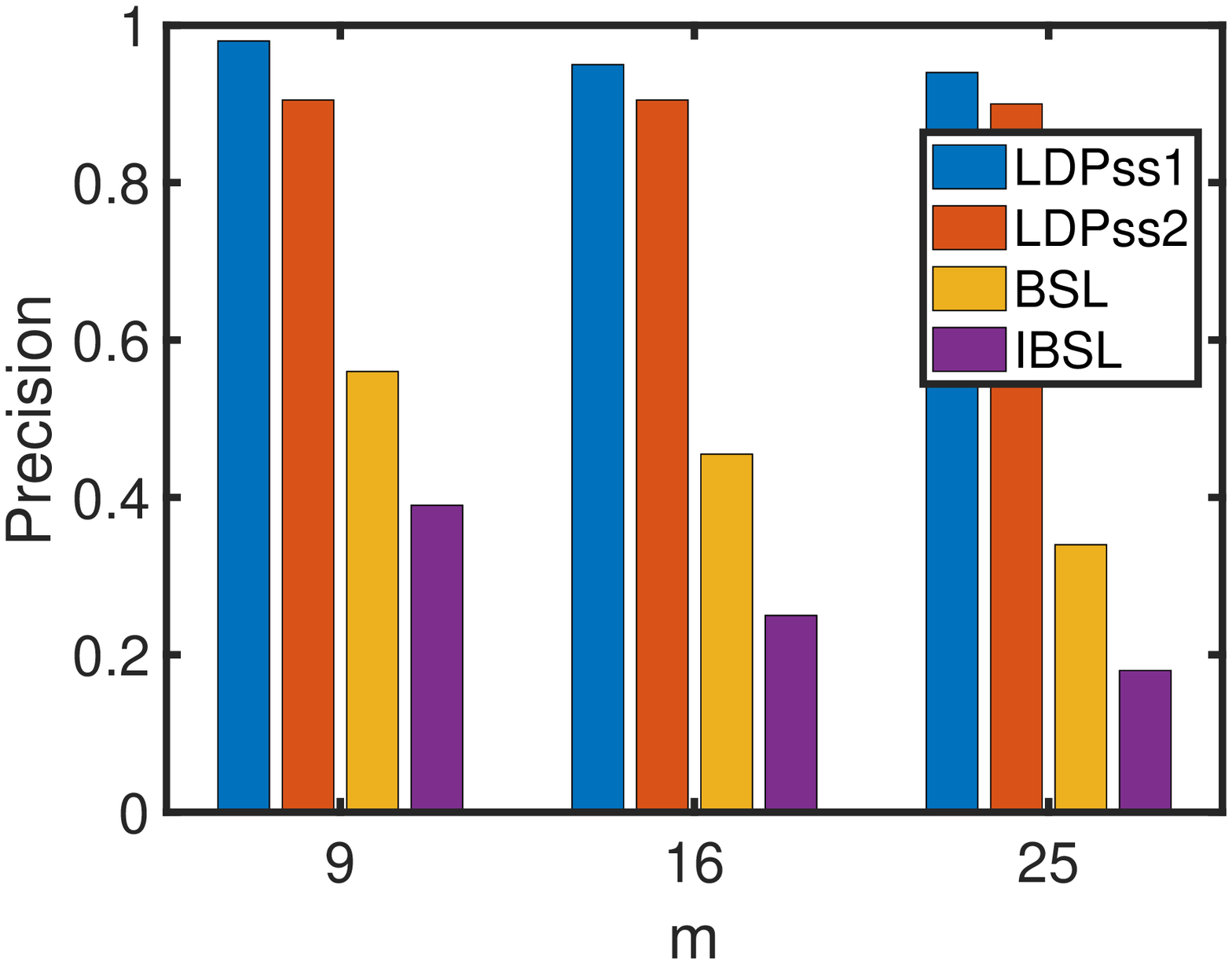}
}
\subfloat[$\epsilon=1$]{
\label{Fig-locationG3}
\includegraphics[scale=0.22]{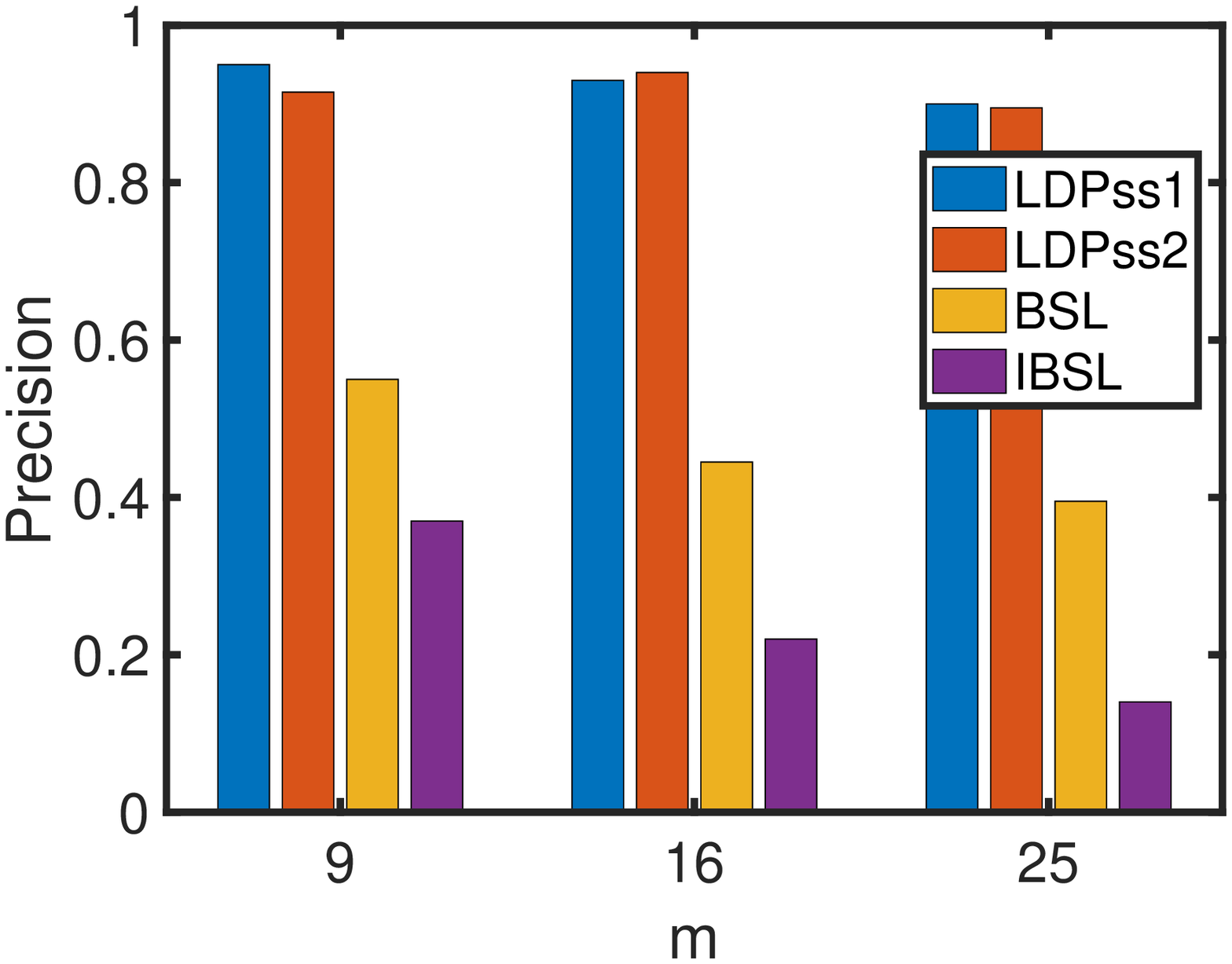}
}
\subfloat[$\epsilon=2$]{
\label{Fig-locationB1}
\includegraphics[scale=0.22]{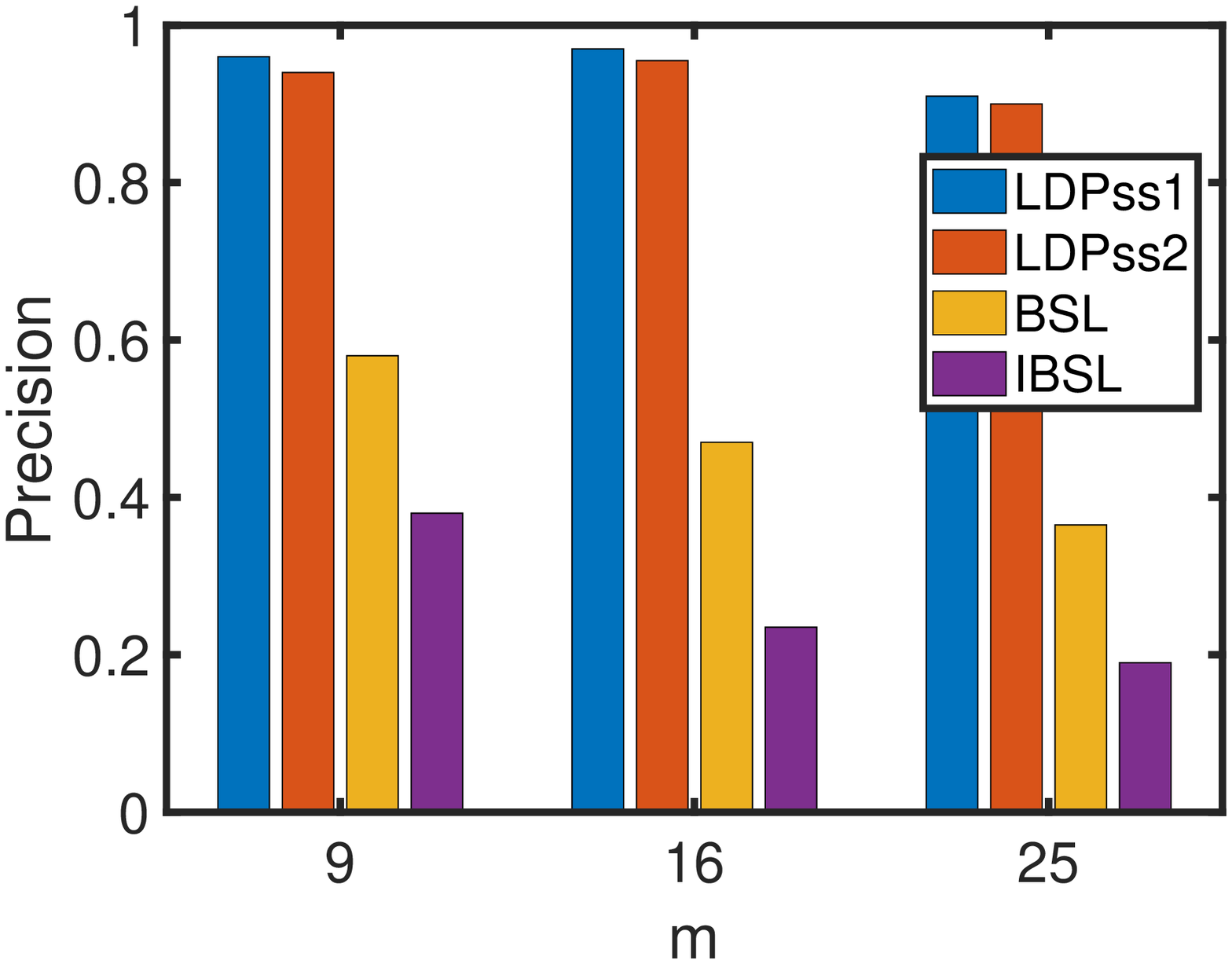}
}

\caption{Performance with different number of locations}
\label{FIG-dimension}
\end{figure*}

\begin{figure*}[htbp]
\centering
\subfloat[Gowalla]{
\label{Fig-alphaG}
\includegraphics[scale=0.22]{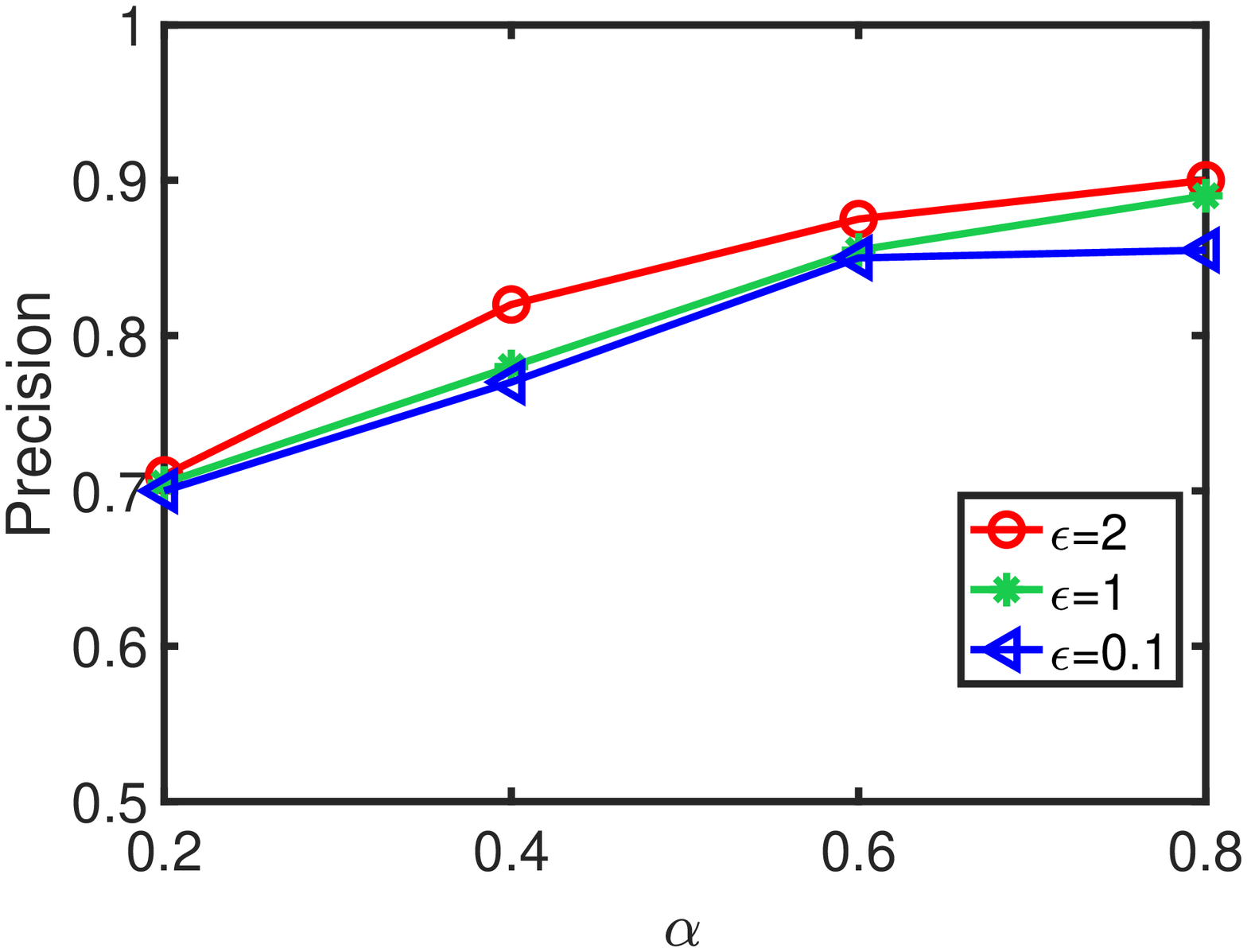}
}
\subfloat[BrigthKite]{
\label{Fig-knnmGres6}
\includegraphics[scale=0.22]{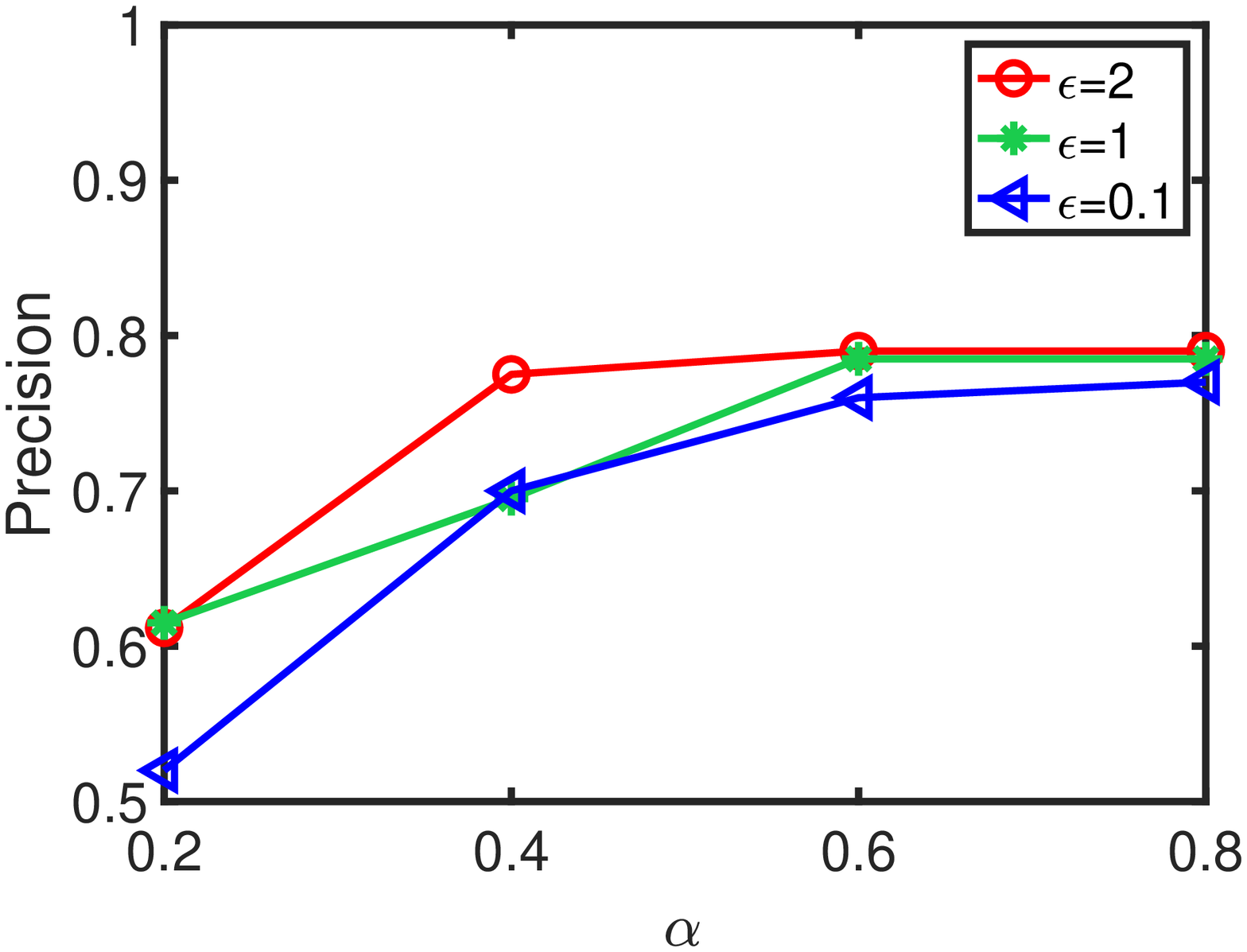}
}
\subfloat[Synthetic]{
\label{Fig-knnmHres6}
\includegraphics[scale=0.22]{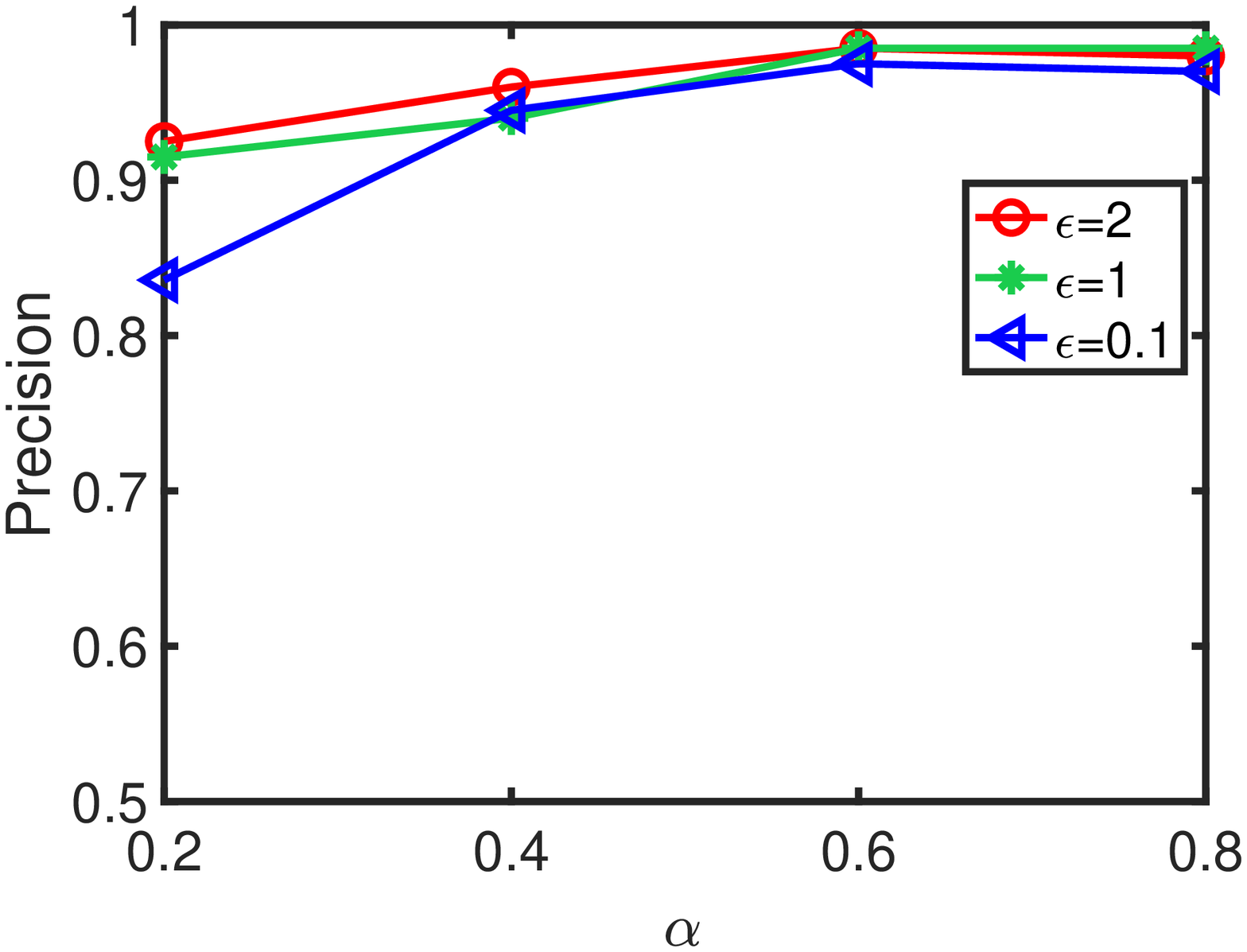}
}
\caption{Effect of $\alpha $ on three datasets}
\label{alpa}
\end{figure*}

Fig. \ref{FIG-numberP} shows the experimental results. 
It is observed that the precision of the proposed methods increase with the increase of the number of participants. 
Specifically, as shown in Fig. \ref{Fig-numberP01}, when $1000$ workers are queried, the
precision achieved by LDPss1 is $0.88$, when the number of workers is increased to $5000$, LDPss1 achieves the precision of $0.98$, which sees an increase of $10\%$. 
LDPss2 achieves a precision at $0.72$ with $3000$ workers and achieves $0.96$ with $5000$ participants.  
The BSL and IBSL follow the same trend as the proposed methods. 
As we can see, the BSL achieves $0.26$ when the scheme involves $1000$ workers with privacy budget $\epsilon=2$ in Fig. \ref{Fig-numberP20}, and the precision achieves $0.32$ and $0.36$ on $3000$ and $5000$ workers respectively. The result supports our expectation, as the difference between number of workers in different locations are bigger with a bigger number of participants, which causes the probability that minimum number of the first $k$ paths smaller than maximum value of the rest to be small for the proposed method and makes the error bound much smaller for both BSL and IBSL.

\subsection{Impact of the number of the locations}

The number of the location indicates the dimension of the dataset. 
We partition the map into $3 \times 3$, $4\times 4$ and $5\times 5$ respectively on Gowalla datasets. And examine the effect of the dimension to the performance of the proposed methods. 

%
%
%

Fig. \ref{FIG-dimension} shows the performance of the methods in terms of precision. 
The dimension of the dataset has a small effect to the proposed method. 
The small fluctuation of the accuracy is caused by the number of population instead of the dimension.
This is because a bigger number of location reduces the number of workers who have the same prefix. 
However, the precision of BSL and IBSL decreases significantly with the increasing of the number of locations. 
Specifically, as shown in Fig. \ref{Fig-locationG}, the LDPss1 achieves a precision at $0.93$ when there are $9$ locations and the precision is still around $0.9$ when the number of location increases to $16$ and $25$. 
However, the BSL achieves around $0.18$ when the location number is $9$, when the location number increases to $16$ and $25$, the precision of BSL achieves $0.15$ and $0.06$ respectively, which sees a decrease by $3\%$ and $12\%$ respectively. A similar result can be found with other privacy budgets. 
In the case of our proposed methods, the use of secret sharing mechanism as a replacement of randomized response for the reporting process eliminates the correlation between the number of location and its precision. However, for the BSL and IBSL, the number of location affects the dimension of the data. The higher dimension means a higher variance, which reduces the accuracy of the statistic result.

\subsection{Performance with different alpha}

%

We found that the method LDPss1 always outperforms LDPss2.  
This is because the LDPss2 conducted double sampling, which reduces the number of functional workers. 
We examine the effect of the sampling probability to LDPss2 by varying the parameter $\alpha$. 
Fig. \ref{alpa} shows the results. 
We varied the number of $\alpha$ between $0.2$ and $1$ in Step $0.2$ on three datasets. 
Smaller $\alpha$ means smaller functional workers are selected to contribute to the statistic. 
We found that LDPss2 has a good performance among all the settings even with a small $\alpha$. 
The precision is increasing with the increase of $\alpha$. 
As shown in Fig. \ref{alpa}, when $\alpha=0.2$, the LDPss2 still can achieve around $0.7$ precision with $\epsilon=0.1$ on Gowalla dataset and achieves $0.84$ when $\alpha=0.8$. 
It has a higher accuracy when the privacy budget $\epsilon$ becomes bigger. 
We can find similar results on both the BrightKite dataset and synthetic dataset.

\section{Related work} \label{RW}
There are a series of works study the problem of frequency estimation following local differential privacy \citep{fanti2016building, wang2019locally,jia2019calibrate,wang2019consistent, zhao2019ldpart}. Among which, three are three types of methods are proposed to deal with high dimensional data statistics. 

\textit{Hash-based method.} 
Erlingsson et al. \citep{erlingsson2014rappor} proposed randomized Aggregatable Privacy-Preserving Ordinal Response (RAPPOR) technology,
which applied the randomized response to the Bloom filters. 
Bassily and Smith \citep{bassily2015local} proposed a random matrix-based method. 
The method is similar to hash-based method. 
The user' value $v\in \mathbb{R}^d $ is mapped to a single bit 
and then the randomized response is performed on this bit. 
To improve the accuracy of the estimation, 
Wang et al. \citep{wang2017locally} proposed an Optimal Local Hasing (OLH) method that hashes the value into $g$ bit, where $g>1$. 
The hash-based method can reduce the query times efficiently by reducing the dimension of the value to a much lower dimension.
However, the collision problem of hash function has to be considered. 
And the decoding process is much more complex. 

\textit{Partition-based method.} The partition-based method tries to identify the frequent values in vector $t$ without going through every bit in $t$. 
Fanti et al. \citep{fanti2016building} proposed to partition the vector into $s$ segments 
and each user only needs to report $2$ segments instead of the whole value. 
The principle is that if a value is frequent, the segment of the value is also frequent. 
Wang et al. \citep{wang2019locally} partition users
into $g$ groups, each user in a group reports a prefix. 
The aggregator estimates the frequent values iteratively.
Partition-based method increases efficiency by reducing the query times.
However, it increases the other computational cost, such as the construction of the candidate set. 
When $g$ is big, the candidate set $C$ is hard to be enumerated. 

\textit{Tree-based method.} 
Bassily et al. \citep{bassily2017practical} proposed a TreeHist protocol which transforms the users' items to binary strings and builds a binary prefix tree. 
The proposed method ensures the number of surviving nodes in each level cannot exceed $O\left(\sqrt{\frac{n}{\log d \log n}}\right)$. 
Wang et al. \citep{wang2018privtrie} proposed an adaptive method to build the tree. 
Specifically, they proposed a novel user candidate set construction method, which defines the available user set $U(v)$ ($v$ refers to the node of the tree) as the set of users who have not participated in the estimation of the support value for any of $v$'s ancestor nodes. 
Tree-based method is not limited to the binary attributes. 
However, there has not been any specific method proposed which reduces the effect of the dimension of the data. This leads to output with high variance and large privacy budget consumption.

\section{Conclusion} \label{C}

In this paper, we considered the private hot path statistic problem 
and proposed a novel solution that enables the workers to participate in crowdsourcing platforms without disclosing their location. 
Specifically, the proposed solution combines the additive secret sharing and local differential privacy technologies. 
The proposed method absorbed the advantages of both secure sharing and local differential privacy, 
which results in accurate statistical results with strict privacy protection and small communication cost. 
We evaluated the performance through extensive experiments 
and the results prove that our method achieves an excellent accuracy compared with the state-of-the-arts.

 \section*{Acknowledgment}
  This research is supported by the National Research Foundation, Prime Minister’s Office, Singapore under its Strategic Capability Research Centres Funding Initiative. 

\appendices




\ifCLASSOPTIONcaptionsoff
  \newpage
\fi



%


\end{document}